\documentclass[runningheads]{llncs}

\usepackage[T1]{fontenc}
\usepackage[utf8]{inputenc}

\usepackage{algorithmic}
\usepackage[ruled,vlined]{algorithm2e}
\usepackage{url}
\usepackage{amsmath}
\usepackage{float}
\usepackage{tikz}
\usepackage{multirow}
\usepackage{color}
\usepackage{listings}
\usepackage{multicol}
\usepackage{enumitem}
\usepackage[framemethod=TikZ]{mdframed}
\usepackage{pgfplots}
\usepackage{pgfplotstable}
\pgfplotsset{compat=1.17}
\usepackage{subcaption}
\usepackage{dingbat}
\usepackage{diagbox} 
\usepackage{hhline}
\usepackage{comment}

\newif\ifExtendedVersion
\ExtendedVersiontrue

\graphicspath{{./figures/}}

\pgfplotsset{
SmallBarPlot/.style={
    font=\small,
    ybar,
    width=\linewidth,
    ymin=0,
    xtick=data,
    xticklabel style={text width=1.5cm, align=center}
},
SmallBarPlotHorizontal/.style={
    font=\small,
    xbar,
    width=\linewidth,
    ytick=data,
    yticklabel style={text width=1.5cm, align=center}
},
BlueBars/.style={
    fill=blue!20, bar width=0.25
},
RedBars/.style={
    fill=red!20, bar width=0.25
},
BlueBars2/.style={
    fill=blue!20, bar width=0.2
},
RedBars2/.style={
    fill=red!20, bar width=0.2
},
GreenBars/.style={
    fill=green!20, bar width=0.2
},
OrangeBars/.style={
    fill=orange!20, bar width=0.2
}
}

\newtheorem{prop}{Proposition}

\DeclareRobustCommand\legendbox[1]{(\textcolor{#1}{#1}~\begin{tikzpicture}[x=0.2cm, y=0.2cm] \draw [color=black, fill=#1!20] (0,0) -- (0,1) -- (0.6,1) -- (0.6,0) -- (0, 0); \end{tikzpicture})}

\pgfplotsset{select coords between index/.style 2 args={
    x filter/.code={
        \ifnum\coordindex<#1\fi
        \ifnum\coordindex>#2\fi
    }
}}

\usepackage[textsize=tiny, textwidth=1.45cm, disable]{todonotes}

\newcommand{\calD}[0]{\mathcal{D}}
\newcommand{\calE}[0]{\mathcal{E}}
\newcommand{\calM}[0]{\mathcal{M}}
\newcommand{\calF}[0]{\mathcal{F}}
\newcommand{\calA}[0]{\mathcal{A}}

\definecolor{mygreen}{rgb}{0,0.6,0}
\definecolor{mygray}{rgb}{0.5,0.5,0.5}
\definecolor{mymauve}{rgb}{0.58,0,0.82}

\renewcommand{\vec}[1]{\ensuremath{\boldsymbol{#1}}}

\usepackage[breaklinks, hidelinks]{hyperref}

\begin{document}


\title{Strategic Remote Attestation:\\Testbed for Internet-of-Things Devices and\\Stackelberg Security Game for Optimal Strategies}
\titlerunning{Strategic Remote Attestation}

    
    

    
\author{Shanto Roy\inst{1} \and
Salah Uddin Kadir\inst{1} \and
Yevgeniy Vorobeychik\inst{2} \and
Aron Laszka\inst{1}}

\institute{University of Houston, Houston, TX 
\and
Washington University in St. Louis, St. Louis, MO}

\maketitle

  \begin{center}
      Published in the proceedings of the 12th Conference on Decision and\\Game Theory for Security (GameSec 2021).
  \end{center}

\begin{abstract}
    Internet of Things (IoT) devices and applications can have significant vulnerabilities, which may be exploited by adversaries to cause considerable harm. 
    An important approach for mitigating this threat is \emph{remote attestation}, 
    which enables the defender to remotely verify the integrity of devices and their software. 
    There are a number of approaches for remote attestation, and each
    has its unique advantages and disadvantages in terms of detection accuracy and computational cost.
    Further, an attestation method may be applied in multiple ways, such as various levels of software coverage.
    Therefore, to minimize both security risks and computational overhead, defenders need to decide strategically which attestation methods to apply and how to apply them, depending on the characteristic of the devices and the potential~losses.
    
    To answer these  questions, we first develop a \emph{testbed for remote attestation of IoT devices}, which enables us to measure the detection accuracy and performance overhead of various attestation methods.
    Our testbed integrates two example IoT applications, memory-checksum based attestation, and a variety of software vulnerabilities that allow adversaries to inject arbitrary code into running applications. 
    Second, we model the problem of finding an optimal strategy for applying remote attestation as a \emph{Stackelberg security game} between a defender and an adversary. 
    We characterize the defender's optimal attestation strategy in a variety of special cases.
    Finally, building on experimental results from our testbed, we evaluate our model and show that optimal strategic attestation can lead to significantly lower losses than na\"ive baseline strategies.
\end{abstract}



\section{Introduction}\label{sec:intro}

With the growing number of Internet of Things (IoT) devices around the world, security has been a significant concern for researchers in the last decade. Due to more exposure in a resource-limited environment,
IoT devices often do not have access to the latest security primitives, and a number of security issues including various software vulnerabilities (e.g., stack and heap-based buffer overflows, format-string vulnerabilities)  exist due to the usage of unsafe languages like C/C++ and vulnerable functions~\cite{xu2018security,steiner2016attestation,parikh2017aslr}. Adversaries can exploit these vulnerabilities to compromise devices by altering the software code or the control flow. Therefore, from a defensive point of view, device attestation that allows an organization to verify integrity remotely is a powerful tool~\cite{steiner2016attestation}.


IoT devices are a preferable target for adversaries these days, and organizations are implementing various methods to mitigate these attacks. 
However, security measures for IoT devices are different from those for servers, since IoT devices usually have low-power resource-limited configurations and are often placed in unknown or unsafe locations. 
To detect and mitigate attacks, a defender may employ remote attestation methods to verify the integrity of a program.

While remote attestation methods can be effective at detecting compromised devices, running attestation can also incur significant computational cost, which may present a prohibitively high overhead on resource-limited IoT devices. 
There are a number of approaches for remote attestation, and each
    has its unique advantages and disadvantages in terms of detection accuracy and computational cost.
    Further, an attestation method may be applied in multiple ways, such as various levels of software coverage.
    Therefore, to minimize both security risks and computational overhead, defenders need to decide strategically which attestation methods to apply and how to apply them, depending on the characteristic of the devices and the potential~losses.

In this paper, we address these questions by (1) implementing an IoT testbed for measuring the detection accuracy and performance overhead of remote-attestation methods and by (2) introducing and solving a game-theoretic model for finding optimal remote-attestation strategies.
Specifically, we formulate and answer the following research questions.
\begin{itemize}
    \item[\textbf{Q1.}] \textbf{Testbed Development:} How to develop an IoT security testbed that can simulate software vulnerability exploitation and evaluate remote attestation?
    \item[\textbf{Q2.}] \textbf{Remote Attestation Methods:} What is the trade-off between the detection rate and computational cost of various remote attestation methods?
    \item[\textbf{Q3.}] \textbf{Optimal Attestation Strategies:} How to model the strategic conflict between a defender and an adversary, and how to find optimal attestation strategies for the defender? 
\end{itemize}

We answer the first question by describing the design and development of our security testbed for IoT device attestation (Section~\ref{sec:design}). We discuss the architecture of our testbed as well as the development of major components, such as vulnerability exploits and attestation methods. 
Our testbed enables us to experiment with software vulnerabilities and exploits and to rigorously evaluate various attestation methods in terms of computational cost and detection rate.

We answer the second question by studying the detection rate and computational cost of memory-checksum based remote attestation (Section~\ref{sec:evaluation}). 
We implement and evaluate {memory-checksum} based attestation  in our testbed for two example IoT applications. We characterize the trade-off between computational cost and detection rate, which we then use to develop the assumptions of our game-theoretic model. 

We answer the third question by developing a Stackelberg security game to model the strategic conflict between a defender and an adversary (Section~\ref{sec:model}). 
We formulate the defender's optimal remote-attestation strategy assuming an adversary who always mounts a best-response attack. 
We show how to compute an optimal strategy in various special cases, and we demonstrate through numerical examples that optimal strategies can attain significantly lower losses than na\"ive baselines. To the best of our knowledge, our model and analysis constitute the first effort to provide optimal remote-attestation strategies. 



\paragraph*{Organization}
The rest of the paper is organized as follows: Section~\ref{sec:background} provides necessary background information. 
Section~\ref{sec:design} discusses the design and development details of our IoT security testbed. 
Section~\ref{sec:model} introduces the attacker-defender model based on Stackelberg security games.
Section~\ref{sec:model_solution} provides analytical results characterizing the defender's optimal attestation strategy. 
Section~\ref{sec:evaluation} presents experimental results from our testbed as well as numerical results on the optimal attestation strategies. 
Section~\ref{sec:related_work} gives a brief overview of related work followed by our concluding remarks and future directions in Section~\ref{sec:conclusion}. 

\section{Background}\label{sec:background}
{ARM} processors are very widely used in IoT platforms. 
Therefore, we develop an ARM-based IoT security testbed to experiment with exploitation and remote attestation on ARM devices. Here, we provide a brief overview of IoT device vulnerabilities, remote attestation methods, and the Stackelberg game model.


\subsection{Software Vulnerabilities and Exploitation in IoT Devices}
Adversaries can take control of an IoT device by hijacking the code execution flow of an application and injecting arbitrary executable code into its memory space. For example, an attacker can use stack- or heap-based \emph{buffer overflow} or \emph{format string vulnerabilities} to inject malicious executable code into a process. 
By injecting executable code, the adversary can alter the functionality of an application (e.g., providing a backdoor to the adversary or causing harm directly). 
While the mitigation for these attacks may be well established in the server and desktop environment, the unique design characteristics of resource-constrained embedded devices makes it challenging to adapt the same defenses techniques. For example, many deeply embedded devices often do not support virtual memory, which is essential for address space layout randomization (ASLR). 


\subsection{IoT Remote Attestation}
Remote attestation establishes trust in a device by remotely verifying the state of the device via checking the integrity of the software running on it. 
Remote attestation methods can be divided into two main categories: hardware and software based.
Hardware-based attestation requires additional dedicated hardware (e.g., Trusted Platform Module) on the device~\cite{abera2016things}. Deploying dedicated hardware can incur additional cost in terms of hardware cost and power consumption, which are often prohibitive for inexpensive or low-power  devices. 
In contrast, software-based attestation requires a \emph{software prover} on the device, which performs specific computations (e.g.,
memory- or time-based checksum~\cite{seshadri2005pioneer,seshadri2004swatt}) and returns the result to the verifier. Note that there are also
hardware-software co-design hybrid platforms for remote attestation~\cite{nunes2019vrased}.
In this paper, we focus on software-based remote attestation.

Steiner et al.\ categorized checksum-based memory attestation in terms of evidence acquisition (software-based, hardware-based, or hybrid), integrity measurement (static or dynamic), timing (loose or strict), memory traversal (sequential or cell/block-based pseudo random), attestation routine (embedded or on-the-fly), program memory (unfilled or filled), data memory (unverified, verified, or erased), and interaction pattern (one-to-one, one-to-many, or many-to-one)~\cite{steiner2016attestation}.
%
Memory checksums can be generated based on \emph{sequential} or \emph{pseudo-random} traversal. 
In sequential traversal, each program memory cell is accessed in a sequential order. 
In contrast, in pseudo-random traversal, memory is accessed in a random cell-by-cell or block-by-block order. The effectiveness of pseudo-random traversal depends on the probability that each cell has been accessed at least~once. 



\subsection{Stackelberg Security Games}
A Stackelberg security game (SSG) is a game-theoretic model, where typically a defending player acts as the leader, and the adversarial player acts as the follower. The leader has the advantage of making the first move, while the follower has the advantage of responding strategically to the leader's move. 
Stackelberg security games have been successfully applied to finding optimal defensive strategies in a variety of settings, both in the cyber and physical domain~\cite{sinha2018stackelberg}.
For example, SSG have helped researchers and practitioners to address a security issues such as security-resource allocation at airports, biodiversity protection, randomized inspections, road safety, border patrol, and so on~\cite{trejo2016adapting,gan2018stackelberg,yin2010stackelberg,bucarey2017building}. 

Game theory can model attacker-defender interactions and characterize optimal strategies given the players' strategy spaces and objectives. 
In our game-theoretic model of remote attestation, the defender acts  as the leader by deciding how often to perform remote attestation, and the adversary acts as the follower by deciding which devices to attack. We provide detailed definitions of the environment, the player's strategy spaces, and their objectives  in Section~\ref{sec:model}.


\section{Testbed Design and Development}\label{sec:design}

In our testbed, multiple IoT applications are running on multiple IoT devices. We implement and enable various software vulnerabilities (e.g., heap-based buffer overflow) in these applications so that adversaries can remotely compromise the devices by exploiting these vulnerabilities. As a result, adversaries can modify the {code} of processes without crashing or restarting them. 
We also integrate memory checksum-based attestation method in the applications. Therefore, a verifier can remotely 
verify the integrity of the vulnerable processes.


\subsection{Testbed Components}
A typical IoT testbed consists of several IoT devices running various IoT server applications. Our testbed also includes two other types of nodes to mount attacks (e.g., code injection) against the IoT devices and to detect the attacks using remote attestation. The architecture of our testbed is presented in Figure~\ref{fig:testbed_archi}.

\begin{figure}[!ht]
    \centering
    \includegraphics[width=.8\textwidth]{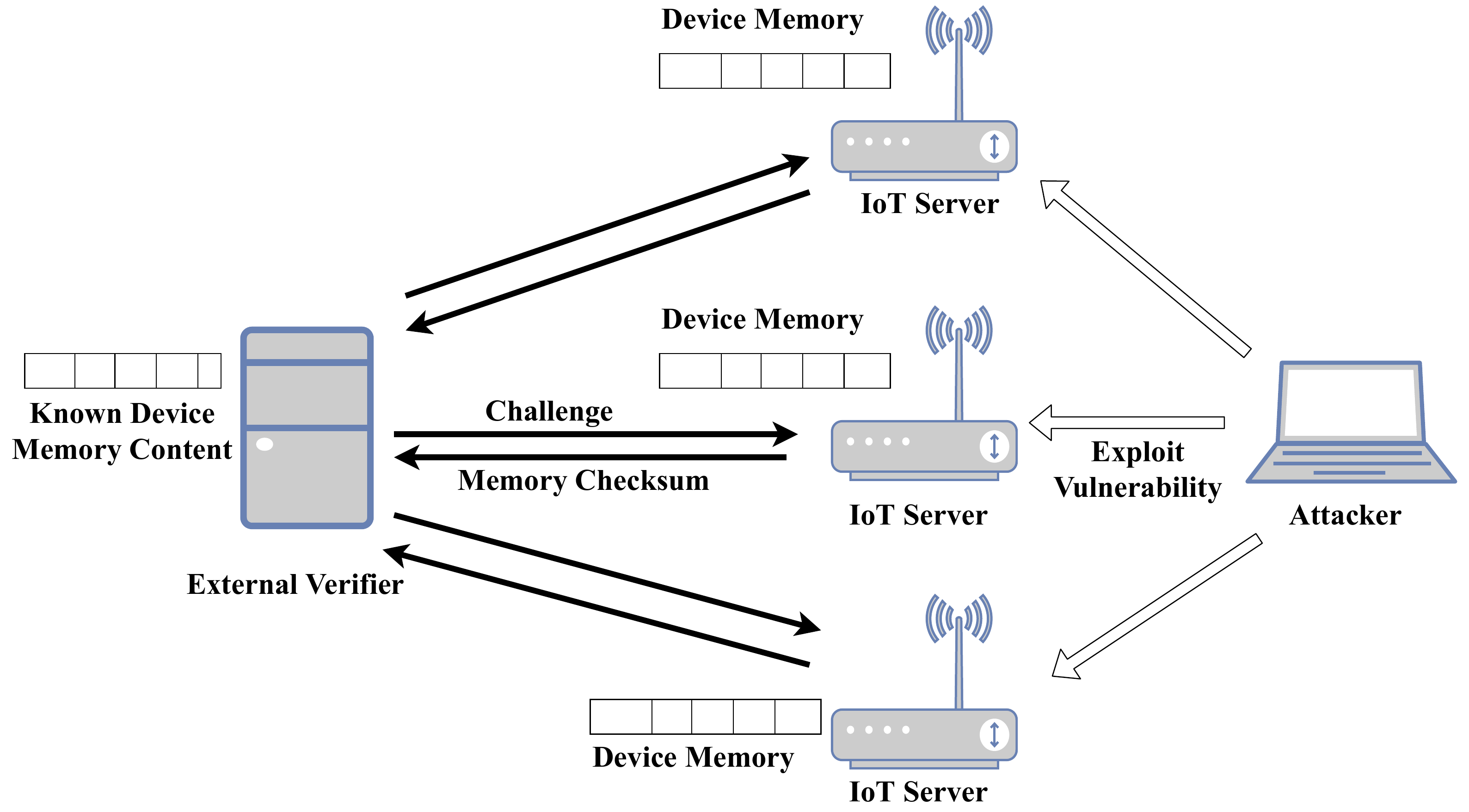}
    \caption{Remote attestation testbed architecture.}
    \label{fig:testbed_archi}
\end{figure}

\paragraph{IoT Server Node}
We set up various IoT server applications on these nodes, add vulnerable code snippets to the applications, and debug them to find exploitation opportunities that can be used to perform code-injection attacks. Then, we incorporate a memory-checksum generator that can calculate a checksum whenever the application receives a challenge from an external verifier node. 

\paragraph{Attacker Node}
The attacker is a  client node that can interact with the IoT application servers and execute various exploits (e.g., stack- or heap-based buffer overflow). The attacker's purpose is to inject or alter the software code of vulnerable applications without crashing the processes. 


\paragraph{External Verifier Node}
The verifier is responsible for performing memory-checksum based attestation of the potentially compromised application servers. For checksum-based attestation, the verifier sends a challenge along with a random seed to the potentially compromised server node and records the response in return to verify. 




\subsection{Testbed Development}
To experiment with various remote attestation strategies, we implement the following features in the testbed: start or terminate various IoT applications, exploit these applications from an attacker node, and generate challenge-response for remote attestation. 

\paragraph{Testbed Setup}
Our testbed uses five Raspberry Pi 3 Model B+ devices (two IoT application servers, an attacker, and two nodes for the verifier).
All devices run Raspbian Linux.
We incorporate two example IoT applications: 
an irrigation server\footnote{\url{https://github.com/NamedP1ayer/IrrigationServer}} and a smart home\footnote{\url{https://github.com/renair/smarthome}}.

\paragraph{Enabling Vulnerabilities}
We disable certain security features of the Linux {kernel} and the {compiler} to enable {stack- and heap}-overflow based exploitation. To enable these vulnerabilities, we disable the ASLR and stack protection; and enable code execution while compiling the applications.

\paragraph{Exploitation Simulation}
We debug all of the applications on the application server nodes to find stack- and heap-based vulnerabilities. Then, we create corresponding exploit payloads on the attacker node. The attacker node sends a request to the server, which triggers the vulnerability and thereby injects a shellcode into the process.


\paragraph{Integrity Verification Simulation}
In our testbed, we implement memory-checksum (sequential- and random-order checksum) as remote attestation strategies, which require an external trusted verifier. 
The verifier can attest a potentially compromised device by sending the same challenge to the target device and an identical isolated  device, and compare their responses.

\section{Game-Theoretic Model of Remote Attestation}\label{sec:model}


Remote attestation enables a defender to detect compromised devices remotely. However, the effectiveness and computational cost of attestation depends on strategies, such as when to attest a device and what method of attestation to employ. 
As IoT devices are resource-limited and attestation incurs computational cost, some devices should not be verified frequently (e.g., devices with low value for an adversary). 
On the other hand, some devices (e.g., ones with high value for the adversary) may need frequent attestation.

To find optimal strategies for remote attestation, we propose a game-theoretic model. Our model is a \emph{two-player, general-sum Stackelberg security game}, where the defender is the leader and the adversary is the follower (i.e., defender first selects its attestation strategy to defend its IoT devices, and then the adversary chooses which devices to attack considering the defender's attestation strategy). 
%
We assume that the defender 
chooses for each device and for each attestation method the probability of applying that method to that device;
the adversary chooses for each device whether to attack it or not.

Table~\ref{tab:game_notations} summarizes the notation of our game-theoretic model.



\subsection{Environment \& Players}


There is a set of IoT devices $\calD$ in the environment, where each individual device $\delta \in \calD$ runs various IoT applications and services. As different devices may have different software stacks, we divide the devices into disjoint classes. These device classes are denoted $\calE_1, \calE_2, \ldots, \calE_n$, where we have $i \neq j \rightarrow \calE_i \cap \calE_j = \emptyset$ and $\bigcup_i \calE_i = \calD$. Devices belonging to the same class have the same vulnerabilities and may be attacked using the same exploits. 

In our Stackelberg game model, there are two players: a defender (leader) and an attacker (follower). The defender tries to minimize the security risks of the IoT devices by detecting  compromises, while the attacker tries to compromise the devices but avoid detection.
To detect compromised devices, the defender uses various attestation methods (e.g., memory checksum, control-flow integrity).
We let $\calM$ denote the set of attestation methods, where each attestation method is an $m \in \calM$. 
If the defender detects a compromised device, the defender resets the device back to its secure state. 

\begin{table}[!ht]
    \setlength{\tabcolsep}{5pt}
    \renewcommand{\arraystretch}{1.15}
    \centering
    \caption{List of Symbols}
    \label{tab:game_notations}
    \begin{tabular}{|c|p{9.5cm}|}
        \hline
        \textbf{Symbol} & \textbf{Description} \\ \hline\hline
        \multicolumn{2}{|c|}{\textbf{Constants}} \\\hline        
        $\calD$ &  set of devices \\\hline
        $\calM$ &  set of attestation method \\\hline
        $\calE_i$ & a set of devices that share common vulnerabilities, where $\calE_i \subseteq \calD$ \\\hline
        $\mu^m$ & detection probability of attestation method $m \in \calM$ when executed on a compromised device \\\hline
        $C_D^m$ &  defender's cost to run attestation method $m \in \calM$ \\\hline
        $C_A^\delta$ &  attacker's cost to compromise device $\delta \in \calD$
        \\\hline
        $C_A^\calE $ & attacker's cost to develop an exploit for a device class $\calE \subseteq \calD$ 
        \\\hline
        $G_D^\delta, G_A^\delta$ & defender's / attacker's gain for compromised device $\delta \in \calD$ \\\hline
        $L_D^\delta, L_A^\delta$ & defender's / attacker's loss for compromised device $\delta \in \calD$ (represented as negative values) \\\hline

        \multicolumn{2}{|c|}{\textbf{Variables}}\\\hline
        $\vec p$ & defender's strategy vector \\\hline
        $\vec a$ & attacker's strategy vector  \\\hline
        $a_\delta$ & attacker's action (i.e., attack or not) against device $\delta \in \calD$ \\\hline
        $p_\delta^m$ & probability of running  attestation method $m \in \calM$ on device $\delta \in \calD$ \\\hline

        \multicolumn{2}{|c|}{\textbf{Functions}}\\\hline            
        $P_\delta(\vec{p})$ & {conditional} probability of defender detecting with strategy $\vec{p}$ that device $\delta \in \calD$ is compromised (given that it is actually compromised) \\\hline 
        $C_D^T(\vec p)$ &  defender's total cost for strategy $\vec{p}$
        \\\hline
        $C_A^T(\vec a)$ &  attacker's total cost
        for strategy $\vec{p}$ \\\hline
        $U_D(\vec p,\vec a)$ & defender's expected utility for strategy profile $(\vec{p}, \vec{a})$ \\\hline
        $U_A(\vec p,\vec a)$ & attacker's expected utility for strategy profile $(\vec{p}, \vec{a})$  \\\hline
        $U_D^\delta(p_\delta^m,a_\delta)$ & defender's expected utility from device $\delta \in \calD$  \\\hline
        $U_A^\delta(p_\delta^m,a_\delta)$ & attacker's expected utility from device $\delta \in \calD$  \\\hline
        $\calF_A(\vec p)$ & attacker's best response against defender strategy $\vec{p}$  \\\hline

        
    \end{tabular}%
\end{table}


\subsection{Strategy Spaces}

Knowing the defender's strategy (i.e., probability of attesting each device using each method), the attacker chooses which devices to attack.
We assume that the attacker follows a deterministic strategy and chooses for each device whether to attack it or not. 
Note that in an SSG, restricting the follower (i.e., the attacker) to deterministic strategies is without loss of generality.
We let the attacker's strategy be represented as a vector $\vec{a} = \langle a_\delta \rangle_{\delta \in \calD}$, 
where 
$a_\delta =1$ means attacking device $\delta \in \calD$, and
$a_\delta =0$ means not attacking device $\delta$. 
Therefore, the attacker's strategy space is
\begin{align*}
\vec{a} \in \{0, 1\}^{|\calD|} .
\end{align*}


On the other hand, the defender can choose a randomized strategy, i.e., 
for each device $\delta \in \calD$ and attestation method $m \in \calM$, the defender chooses the probability $p_\delta^m \in [0, 1]$ of running attestation method $m$ on device $\delta$. 
We let the defender's strategy be represented as a vector
$\vec p = \langle p_\delta^m \rangle_{\delta \in \calD, m \in \calM}$, where $p_\delta^m = 0$ means never running method $m \in \calM$ on device $\delta \in \calD$, and $p_\delta^m = 1$ means always running method $m$ on device $\delta$. 
Therefore, the defender's strategy space is
\begin{align*}
 \vec{p} \in [0, 1]^{|\calD \times \calM|} .
\end{align*}


\subsection{Utility Functions}
Next, we formalize the players' objectives by defining their utility functions. 

\paragraph{Defender's Utility}

Different attestation methods can have different detection rates (i.e., different probability of detecting an attack when the method is run on a compromised device).
For each attestation method  $m \in \calM$, we let $\mu^m$ denote the probability that method $m$ detects that the device is compromised.


However, the defender can run multiple attestation methods on the same device, and any one of these may detect the compromise. Therefore, the probability $P_\delta\left(\vec{p}\right)$ of detecting that device $\delta \in \calD$ is compromised when the defender uses attestation strategy $\vec{p}$ is
\begin{align}
    P_\delta\left(\vec{p}\right) = 1 - \prod_{m \in \calM} \left(1 - \mu^m \cdot p_\delta^m\right) . \label{eqn:accuracy}
\end{align}

Each attestation method also has a computational cost, which the defender incurs for running the method on a device.
For each attestation method $m \in \calM$,
we let $C_D^m$ be the cost of running method $m$ on a device. 
Then, the defender's expected total cost $C_D^T \left(\vec p\right)$ for running attestation following strategy $\vec{p}$ is
\begin{align}
    C_D^T \left(\vec p\right) &= \sum_{\delta \in \calD}\sum_{m \in \calM} C_D^m \cdot p_\delta^m . \label{eqn:defender_cost} 
\end{align}
Note that
the expected total cost of attestation $C_D^T(\vec p)$ 
depends on the probability of running attestation (higher the probability $p_\delta^m$, higher the expected cost for device $\delta$ and method $m$). 

Next, we let $G_D^\delta$ be the defender's gain when the attacker chooses to attack device $\delta \in \calD$ and the defender detects that the device is compromised. 
On the other hand, let $L_D^\delta$ be the defender's loss when  the attacker chooses to attack device $\delta$ and the defender does not detect that the device is compromised.
Then, we can express the defender's expected utility $U_D(\vec p,\vec a)$ when the defender uses attestation strategy $\vec{p}$ and the attacker uses attack strategy $\vec{a}$ as
\begin{align}
    U_D(\vec p,\vec a)
    = \sum_{\delta \in \calD}&  \left[ G^\delta_D \cdot P_\delta(\vec{p}) + L^\delta_D \cdot (1-P_\delta(\vec{p}))\right]\cdot a_\delta - C_D^T(\vec p) .  \label{eqn:defender_utility}
\end{align}

\paragraph{Attacker's Utility }
Let $C_A^\calE$ be the cost of developing an exploit for device class~$\calE$,
and let $C_A^\delta$ be the cost of attacking a particular device $\delta \in \calD$.
For any attack strategy $\vec{a}$,
the set of device classes that the adversary attacks can be expressed as $\{ \calE \,|\, \exists \, \delta \in \calE \, ( a_\delta = 1) \}$.
Then, we can express the adversary's total cost $C_A^T(\vec a)$ for attack strategy $\vec{a}$ as
\begin{align}
    C_A^T(\vec a) &= \sum_{\calE} \left ( C_A^\calE \cdot 1_{\left\{\exists \delta \in \calE ( a_\delta = 1 )\right\}} + \sum_{\delta \in \calE} C_A^\delta \cdot a_\delta \right ) . \label{eqn:attacker_cost_2}
\end{align}
Note that the attacker incurs cost for both developing an exploit for each class that it targets as well as for each individual device.

Similar to the defender, we let attacker's gain and loss for attacking a device $\delta \in \calD$ be $G_A^\delta$ and $L_A^\delta$ when the compromise is not detected and detected, respectively.
Then, we can express the adversary's expected utility $U_A(\vec p,\vec a)$ when the defender uses attestation strategy $\vec{p}$ and the attacker uses attack strategy $\vec{a}$~as
\begin{align}
    U_A(\vec p, \vec a) = \sum_{\delta \in \calD} \left[ L_A^\delta \cdot P_\delta(\vec{p}) + G_A^\delta \cdot (1-P_\delta(\vec{p})) \right] \cdot a_\delta  - C_A^T(\vec a) . \label{eqn:attacker_utility}
\end{align}
For the sake of simplicity, we assume that with respect to gains and losses from compromises, the players' utilities are zero sum, that is, 
$G_D^\delta = -L_A^\delta$ and $L_D^\delta = -G_A^\delta$.
Note that the game is not zero sum due to the players' asymmetric costs $C_D^T(\vec p)$ and $C_A^T(\vec a)$.

\subsection{Solution Concept}
We assume that both the defender and attacker aim to maximize their expected utilities.
To formulate the optimal attestation strategy for the defender, we  first define the attacker's best-response strategy.

In response to a defender's strategy $\vec{p}$, the attacker always chooses an attack strategy $\vec{a}$ that maximizes the attacker's expected utility $U_A(\vec p, \vec a)$.
Therefore, we can define the {attacker's best response} as follows.

\begin{definition}[Attacker's best response]
Against a defender strategy $\vec{p}$, 
the attacker's \emph{best-response strategy} $\calF_A(\vec p)$ is
\begin{align}
    \calF_A(\vec p) = \operatorname{argmax}_{\vec a} U_A(\vec p, \vec a) . \label{eqn:attackersoptimal}
\end{align}
\end{definition}

Note that the best response is not necessarily unique (i.e., $\calF_A$ may be a set of more than one strategies).
Hence, as is usual in the literature, we will assume tie-breaking in favor of the defender to formulate the optimal attestation strategy.

Since the defender's objective is to choose an attestation strategy $\vec{p}$ that maximizes its expected utility $U_D(\vec p, \vec a)$ anticipating that the attacker will choose a best-response strategy from $\calF(\vec{p})$,
we can define the defender's optimal strategy as follows.

\begin{definition}[Defender's optimal strategy]
%
The defender's \emph{optimal attestation strategy} $\vec{p}^*$ is
\begin{align}
 \vec{p}^* = \operatorname{argmax}_{\vec p,\, \vec a \in \calF_A(\vec p)} U_D(\vec p, \vec a) \label{eqn:defendersoptimal}
\end{align}
\end{definition}

\section{Analysis of Optimal Attestation Strategies}\label{sec:model_solution}

Here, we present analytical results on our game-theoretic model, characterizing the defender's optimal strategy in important special cases.
For ease of exposition, we present these special cases in increasing generality.
\ifExtendedVersion
We provide the proofs in Appendix~\ref{sec:proofs}.
\fi

\subsection{Case 1: Single Device and Single Attestation Method}


First, we assume that there exists only one device $\delta$ and one attestation method~$m$.

\paragraph*{Attacker's Best-Response Strategy}
%
Whether the attacker's best response is to attack or not depends on the defender's strategy ${p}_\delta^m$.
Further, it is easy to see that if attacking is a best response for some ${p}_\delta^m$, then it must also be a best response for any $\hat{p}_\delta^m < {p}_\delta^m$.
Therefore, there must exist a threshold value $\tau_\delta$ of the defender's probability ${p}_\delta^m$ that determines the attacker's best response.

\begin{lemma} \label{lem:single_device_single_class}
The attacker's best-response strategy $\calF(\vec p)$ is
\begin{align}
    \calF(\vec p) = \begin{cases}
    \{1\} & \text{ if $p_\delta^m < \tau_\delta$} \\
    \{0, 1\} & \text{ if $p_\delta^m = \tau_\delta$} \\
    \{0\} & \text{ otherwise,}
    \end{cases} \label{eqn:final_single_device_single_attestation_condition}
\end{align}
where
\begin{equation}
    \tau_\delta = \frac{1}{\mu^m} \cdot \frac{C_A^{\calE} + C_A^\delta - G_A^\delta}{L_A^\delta - G_A^\delta} .
\end{equation}
\end{lemma}

In other words, it is a best response for the attacker to attack if the defender's attestation probability $p_\delta^m$ is lower than the threshold $\tau_\delta$; and it is a best response not to attack if the probability $p_\delta^m$ is higher than the threshold $\tau_\delta$.





\paragraph*{Defender's Optimal Strategy}
The defender may pursue one of two approaches for maximizing its own expected utility: selecting an attestation probability that is high enough to deter the attacker from attacking (i.e., to eliminate losses by ensuring that not attacking is a best response for the attacker); or selecting an attestation probability that strikes a balance between risk and cost, accepting that the adversary might~attack.

First, from Equations~\eqref{eqn:defender_cost} and~\eqref{eqn:final_single_device_single_attestation_condition}, it is clear that the lowest-cost strategy for deterring the attacker is $p_\delta^m = \tau_\delta$.
Second, if the defender does not deter the attacker, then it must choose a probability $p_\delta^m$ from the range $\left[0, \tau_\delta\right]$ that maximizes $U_D(p_\delta^m, 1)$.
Then, it follows from Equation~\eqref{eqn:defender_utility} that the optimal probability is either $p_\delta^m = 0$ or  $\tau_\delta$, depending on the constants $\mu^m, C_D^m, G^\delta_D,$ and  $L^\delta_D$.

\begin{proposition}\label{prop:single_device_single_class}
The defender's optimal attestation strategy ${p^*}_\delta^m$~is 
\begin{align}
    {p^*}_\delta^m
    = \begin{cases} 0 & \text{if } C_D^m \geq (G^\delta_D  - L^\delta_D)\cdot \mu^m \text{ and } \tau_\delta \geq \frac{L_D^\delta}{- C_D^m} \\
    \tau_\delta & otherwise. \end{cases}
\end{align}
\end{proposition}

Note that the first case corresponds to when deterrence is not necessarily better than non-deterrence (first condition), and  for non-deterrence strategies, minimizing risks over costs is not better (second condition).



\subsection{Case 2: Multiple Devices and Single Device Class}


Next, we generalize our analysis by allowing multiple devices $\calD$, but assuming a single device class $\calE = \calD$ and single attestation method $m$.

\paragraph*{Attacker's Best-Response Strategy}
First, the attacker needs to decide whether it will attack at all: if the attacker does not attack at all, it attains $U_A(\vec{p}, \vec{0}) = 0$ utility; if the attacker does attack some devices, it incurs the cost $C_A^{\mathcal{E}}$ of attacking the class once, and it will need to make decisions for each individual device $\delta \in \calD$ {without considering} this cost $C_A^{\mathcal{E}}$.
The latter is very similar to \emph{Case 1}  since for each individual device $\delta$, the decision must be based on a threshold value $\tau_\delta$ of the attestation probability $p_\delta^m$; however, this threshold must now ignore $C_A^{\mathcal{E}}$.

\begin{lemma}
\label{lem:multi_dev_sing_class}
The attacker's best-response strategy $\calF(\vec{p})$ is
\begin{equation}
\calF(\vec{p}) = \begin{cases} \left\{\vec{a}^*\right\} & \text{ if } U_A\left(\vec{p}, \vec{a}^*\right) > 0 \\
\left\{\vec{a}^*, \vec{0} \right\} & \text{ if } U_A\left(\vec{p}, \vec{a}^*\right) = 0 \\
\left\{\vec{0} \right\} & \text{ otherwise,}
\end{cases}
\end{equation}
where
\begin{equation}
{a}^*_\delta = \begin{cases} 1 & \text{ if } p_\delta^m < \overline{\tau}_\delta \\
0 & \text{ otherwise,}
\end{cases} \label{eq:bestrespifattack}
\end{equation}
and
\begin{equation}
  \overline{\tau}_\delta =  \frac{1}{\mu^m} \cdot \frac{C_A^\delta - G_A^\delta}{L_A^\delta - G_A^\delta} .
\end{equation}
\end{lemma}

Note that strategy $\vec{a}^*$ is a utility-maximizing strategy for the attacker assuming that it has already paid the cost $C_A^{\mathcal{E}}$ for attacking the class.
Hence, the decision between attacking (in which case $\vec{a}^*$ is optimal) and not attacking at all ($\vec{a} = \vec{0}$) can is based on the utility $U_A(\vec{p}, \vec{a}^*)$ obtained from strategy $\vec{a}^*$ and the utility $U_A(\vec{p}, \vec{0}) = 0$ obtained from not attacking at all.





\paragraph*{Defender's Optimal Strategy}
Again, the defender must choose between deterrence and acceptance (i.e., deterring the adversary from attacking or accepting that the adversary might attack).
However, in contrast to \emph{Case 1}, the defender now has the choice between completely deterring the adversary from attacking (i.e., adversary is not willing to incur cost $C_A^\calE$ and hence attacks no devices at all) and  deterring the adversary only from attacking some devices (i.e., adversary incurs cost $C_A^\calE$ and attacks some devices, but it is deterred from attacking other devices). 

\begin{proposition}
\label{prop:multi_dev_sing_class}
The defender's optimal attestation strategy $\vec{p}^*$ is
\begin{equation}
\vec{p}^* = \begin{cases}
\left\{ \vec{p}^\text{ND} \right\} & \text{ if } U_D\left(\vec{p}^\text{ND}, \vec{a}^*\right) > U_D\left(\vec{p}^\text{D}, \vec{0}\right) \\
\left\{ \vec{p}^\text{ND}, \vec{p}^\text{D} \right\} & \text{ if } U_D\left(\vec{p}^\text{ND}, \vec{a}^*\right) = U_D\left(\vec{p}^\text{D}, \vec{0}\right) \\
\left\{ \vec{p}^\text{D} \right\} & \text{ otherwise,}
\end{cases}    
\end{equation}
where
\begin{align}
    \left(p^\text{ND}\right)_\delta^m
    = \begin{cases} 0 & \text{if } C_D^m \geq (G^\delta_D  - L^\delta_D)\cdot \mu^m \text{ and } \overline{\tau}_\delta \geq \frac{L_D^\delta}{- C_D^m} \\
    \overline{\tau}_\delta & otherwise, \end{cases}
\end{align}
$\vec{a}^*$ is as defined in Equation~\eqref{eq:bestrespifattack} with $\vec{p} = \vec{p}^\text{ND}$,
and 
\begin{align}
\label{eq:det_opt}
   {\vec{p}^\text{D}} = \operatorname{argmin}_{\left\{\vec p:\, U_A(\vec p, \vec 1) \leq 0 \, \wedge \, \forall \delta \left( p_\delta^m \in [0,\tau_\delta] \right) \right\}}  \sum_{\delta \in \calD} C_D^m \cdot p_\delta^m . 
\end{align}
\end{proposition}

Note that $\vec{p}^\text{ND}$ is the optimal attestation strategy if  the defender does not completely deter the adversary from attacking, calculated similarly to \emph{Case 1}; $\vec{p}^\text{D}$ is the optimal attestation strategy if the defender completely deters the adversary from attacking, which may be computed by solving a simple linear optimization (Equation~\ref{eq:det_opt}).

\subsection{Case 3: Multiple Devices and Multiple Device Classes}

Next, we consider multiple device classes $\calE_1, \calE_2, \ldots, \calE_n$.
We generalize our previous results by observing that both the attacker's and defender's decisions for each class of devices are independent of other classes.


\begin{lemma} \label{lem:multiple_device_multiple_class}
For each device class $\mathcal{E}_i$, let $\vec{a}_i$ be a best response as given by Lemma~\ref{lem:multi_dev_sing_class}.
Then, $\langle \vec{a}_1, \vec{a}_2, \ldots, \vec{a}_n\rangle$ is a best-response attack strategy.
\end{lemma}


\begin{proposition} \label{prop:multiple_device_multiple_class}
For each device class $\mathcal{E}_i$, let $\vec{p}^*_i$ be an optimal attestation strategy as given by Proposition~\ref{prop:multi_dev_sing_class}.
Then, $\langle \vec{p}^*_1, \vec{p}^*_2, \ldots, \vec{p}^*_n\rangle$ is an optimal attestation strategy.
\end{proposition}

\section{Numerical Results}\label{sec:evaluation}
Here, we present experimental results from our testbed, which confirm our modeling assumptions, as well as numerical results on our game-theoretic model.

\subsection{Experimental Results from the Remote Attestation Testbed}

We consider an experimental setup with two test applications, an irrigation  and a smart\-home application.
We implement sequential and pseudo-random memory checksum as exemplary software-based remote attestation methods. 
%
Software-based attestation incurs various costs; in this section, we study checksum-based remote attestation  in terms of memory and computational overhead. 
We also evaluate checksum-based attestation in terms of detection rate.




\begin{figure}[t]
\centering
\pgfplotstableread[col sep=comma,]{Data/irrigation_server.csv}\dataA
\pgfplotstableread[col sep=comma,]{Data/smarthome_server.csv}\dataB
\pgfplotsset{scaled x ticks=false}
\begin{tikzpicture}

\begin{axis}[
    ymin=-5,
    ymax = 110,
    xmin=-5,
    xmax = 110,
	xlabel={Checksum Memory Coverage [\%]},
	ylabel={Detection Rate [\%]},
	width=0.55\textwidth,
    legend style={at={(0.01,.88)},anchor=west}]
    \addplot [color=red, mark=x] table [x=covered_memory_size, y=detection_rate, col sep=comma]{\dataA};
    \legend{irrigation}
    \addplot [color=blue, mark=x] table [x=covered_memory_size, y=detection_rate, col sep=comma]{\dataB};
    \addlegendentry{smarthome}
\end{axis}

\end{tikzpicture}
\caption{Detection rate of pseudo-random memory checksum as a function of memory coverage.}
\label{fig:detectionvsblock}
\end{figure}
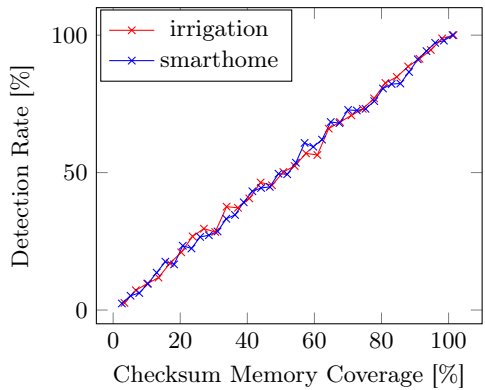

\begin{figure*}[!t]
\centering
\begin{subfigure}{.48\linewidth}
\centering
\includegraphics[width=\textwidth]{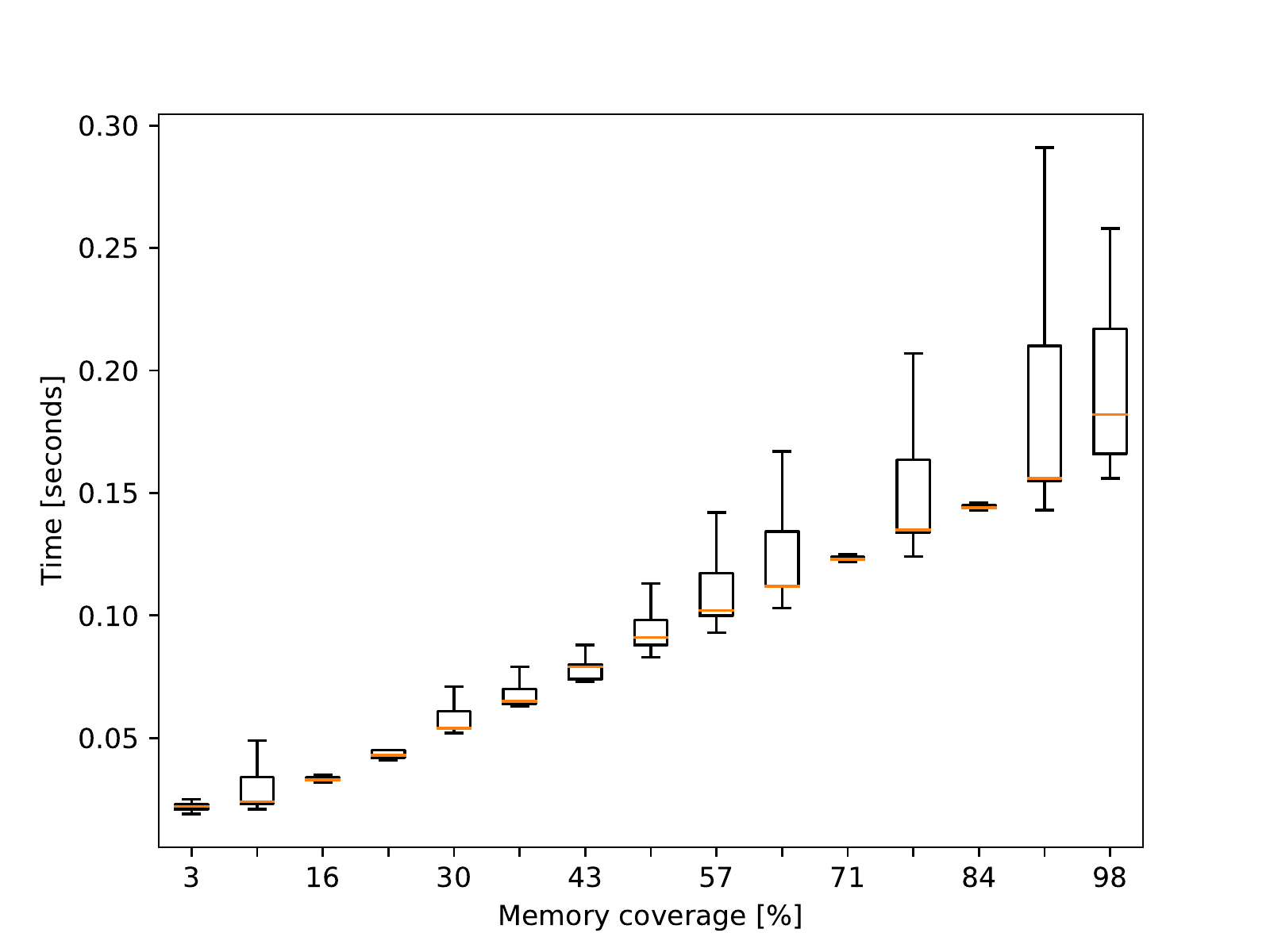}
\caption{irrigation }
\label{fig:irrigation_spent_time}
\end{subfigure} \hspace{.3em}
\begin{subfigure}{.48\linewidth}
\centering
    \includegraphics[width=\textwidth]{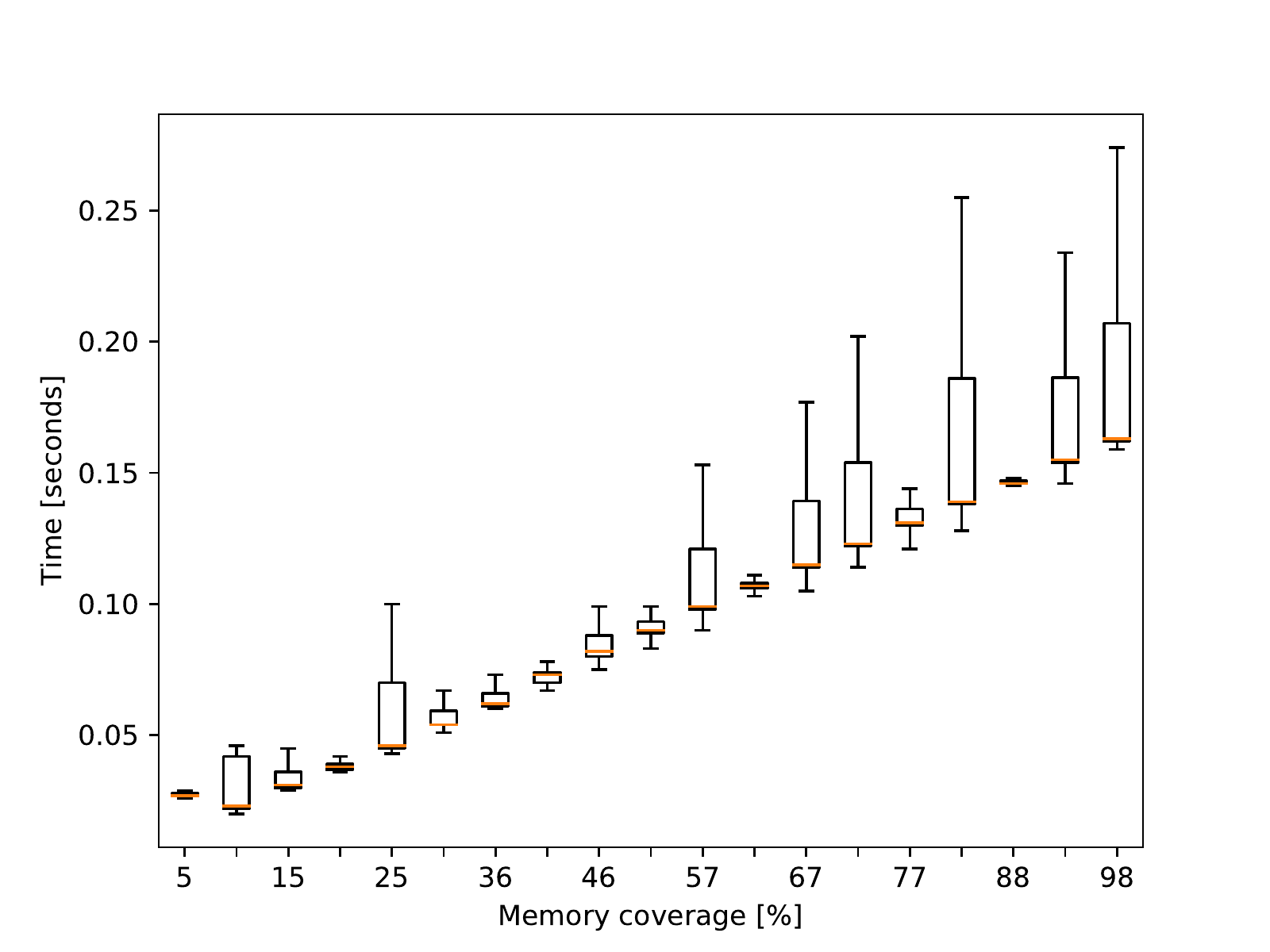}
    \caption{smarthome}
    \label{fig:smarthome_spent_time}
\end{subfigure}

\caption{Running time of checksum calculation as a function of memory coverage.}
\label{fig:iot_projects_comparison}
\end{figure*}

\paragraph{Detection Rate of Pseudo-random Memory Checksum}
In this experiment, we study the efficacy of pseudo-random memory checksum in terms of detecting changes to the code of a running application. We use a block-based pseudo-random technique, where each block is $500$ bytes. We start our experiment with checking $200$ blocks, which are selected pseudo-randomly based on a seed value. Then, we increase the number of blocks by $200$ in iterations to measure the impact of increasing memory coverage. In each iteration, we run  the pseudo-random memory checksum $500$ times, using a different random seed each time, to reliably measure the detection rate for a certain number of blocks.


Figure~\ref{fig:detectionvsblock} shows the
detection rate of pseudo-random memory checksum as a function of the fraction of memory covered by the checksum calculation, for our two test applications. 
Note that we express the fraction of memory covered as a percentage. Specifically, we calculate the ratio as $(\text{number of blocks} \times \text{block size} \times 100) / \text{total memory  size of the program}$.
We find that detection rate increases roughly proportionally with memory coverage, ranging from 0\% to 100\%, which supports our modeling choices. 



\paragraph{Running Time of Pseudo-Random Memory Checksum}
Next, we study the running time of calculating pseudo-random memory checksum with memory coverage ranging from 3\% to 98\%. 
For each memory-coverage level, we run the checksum calculation 500 times to obtain reliable running-time measurements.
Figure~\ref{fig:iot_projects_comparison} shows the distribution of running time for various memory-coverage level for our two test applications. 
We find that similar to detection rate, the average of running time also increases proportionally with memory coverage, which supports our modeling choices. 




\subsection{Evaluation of Game-Theoretic Model and Optimal Strategies}

To evaluate our model and optimal strategies, we 
consider an example environment consisting of $|\calD|=50$ IoT devices from $5$ different classes ($10$ devices in each class $\calE_i$), and for simplicity, we consider a single attestation method $m$ implemented in these devices.
For each device $\delta$, we choose both the defender's and the attacker's gain values  $G_D^\delta$ and $G_A^\delta$ uniformly at random from $[20, 40]$.
We assume that the game is zero-sum with respect to gains and losses; that is, we let the players' losses ($L_D^\delta, L_A^\delta$) be $G_D^\delta = -L_A^\delta$ and $L_D^\delta = -G_A^\delta$. 
Finally, we choose the detection probability of the attestation method $\mu$ uniformly at randomly from $[0.5, 0.9]$, the attestation cost $C_D$ from $[0, 10]$, the exploit development cost $C_A^\calE$ from $[15, 40]$, and the device attack costs $C_A^\delta$ from $[1, 3]$ for each device $\delta$. 

\begin{figure}[t]
\pgfplotstableread[col sep = comma]{optimal_utilities.csv}\table
\centering
\begin{tikzpicture}
    \begin{axis}[
            xbar,
            legend style={at={(.6,.91)},anchor=west},
            SmallBarPlot,
            width=0.7\textwidth,
            height=0.4\textwidth,
            xticklabels={$\vec{p}\!=\!\vec{0}$, $\vec{p}\!=\!\vec{1}$, uniform $\vec{p}$, \mbox{optimal $\vec{p}^*$}},
            ylabel=Utilities,
            enlarge x limits={abs=0.5},
            ymin=-1100,
        ]
        \addplot [BlueBars] table [x expr=\coordindex, y=defender] {\table};
        \addplot [RedBars] table [x expr=\coordindex, y=attacker] {\table};
    \end{axis}
    \end{tikzpicture}
    \caption{Comparison between optimal and na\"ive defender strategies based on the defender's utility \legendbox{blue} and the attacker's utility \legendbox{red}, assuming that the attacker chooses its best response.}
    \label{fig:optimal_utilities}
\end{figure}
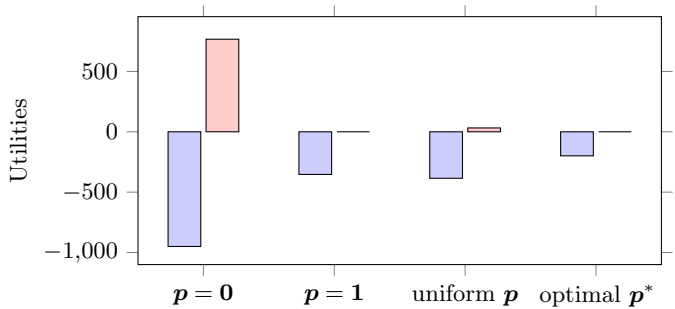

\paragraph{Comparison to Na\"ive Baselines}

We compare  the defender's optimal attestation strategy $\vec{p}^*$ to three na\"ive baseline strategies: $\vec{p}=\vec{0}$, $\vec{p}=\vec{1}$, and an optimal uniform $\vec{p}$ (i.e., same probability $p_\delta^m$ for all devices $\delta$, but this probability is chosen to maximize the defender's utility given that the adversary always chooses its best response). 
Figure~\ref{fig:optimal_utilities} shows the players' utilities for the optimal and na\"ive defender strategies, assuming that  the adversary chooses its best response in each case.
We see that the optimal strategy outperforms the na\"ive baselines in terms of the defender's utility.
Specifically, it outperforms $\vec{p} = \vec{0}$ and optimal uniform $\vec{p}$ by deterring the adversary from attacking, which these na\"ive baselines fail to achieve; and it outperforms $\vec{p} = \vec{1}$ by achieving deterrence at a lower cost.


\begin{figure*}[t]
\pgfplotstableread[col sep = comma]{attacker_attacks.csv}\table
\centering
\begin{subfigure}{.47\linewidth}
\centering
\begin{tikzpicture}
    \begin{axis}[
            xbar,
            legend style={at={(.6,.91)},anchor=west},
            SmallBarPlot,
            xticklabels={$\vec{p}\!=\!\vec{0}$, $\vec{p}\!=\!\vec{1}$, uniform $\vec{p}$, \mbox{optimal $\vec{p}^*$}},
            xticklabel style={rotate=24, font=\footnotesize},
            ylabel=Utilities,
            enlarge x limits={abs=0.5},
            ymin=-1100,
        ]
        \addplot [BlueBars] table [x expr=\coordindex, y=defender] {\table};
        \addplot [RedBars] table [x expr=\coordindex, y=attacker] {\table};
    \end{axis}
    \end{tikzpicture}
    \caption{Attacker chooses to attack}
    \label{fig:attestation_probability_comparison_a}
\end{subfigure} \hspace{1em}
\begin{subfigure}{.47\linewidth}
\pgfplotstableread[col sep = comma]{attacker_not_attack.csv}\table
\centering
    \begin{tikzpicture}
    \begin{axis}[
            xbar,
            legend style={at={(.6,.91)},anchor=west},
            SmallBarPlot,
            xticklabels={$\vec{p}\!=\!\vec{0}$, $\vec{p}\!=\!\vec{1}$, uniform $\vec{p}$, \mbox{optimal $\vec{p}^*$}},
            xticklabel style={rotate=24, font=\footnotesize},
            yticklabel pos=right,
            ylabel=Utilities,
            enlarge x limits={abs=0.5},
            ymin=-500,
        ]
        \addplot [BlueBars] table [x expr=\coordindex, y=defender] {\table};
        \addplot [RedBars] table [x expr=\coordindex, y=attacker] {\table};
    \end{axis}
    \end{tikzpicture}
    \caption{Attacker chooses not to attack}
    \label{fig:attestation_probability_comparison_b}
\end{subfigure}

\caption{Detailed comparison between optimal and na\"ive defender strategies based on the defender's utility \legendbox{blue} and the attacker's utility \legendbox{red}.}
\label{fig:attestation_probability_comparison}
\end{figure*}
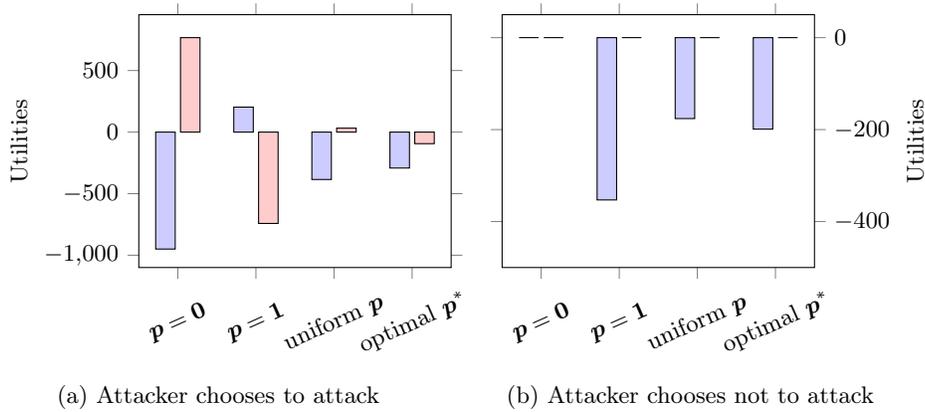

\paragraph{Detailed Comparison to Na\"ive Baselines}
Figure~\ref{fig:attestation_probability_comparison} provides a more detailed comparison between the optimal attestation strategy $\vec{p}^*$ and the three na\"ive baselines. 
In contrast to Figure~\ref{fig:optimal_utilities}, this figure shows utilities both in the case when the adversary decides to attack (Figure~\ref{fig:attestation_probability_comparison_a}) and in the case when it decides to not attack at all (Figure~\ref{fig:attestation_probability_comparison_b}).
In Figure~\ref{fig:attestation_probability_comparison_a}, we see that the adversary can obtain a positive utility from attacking against $\vec{p} = \vec{0}$ and the optimal uniform~$\vec{p}$. 
Therefore, these strategies do not deter the adversary from attacking. 
In contrast, the adversary's utility is negative against both $\vec{p} = \vec{1}$ and the optimal strategy $\vec{p}^*$.
In Figure~\ref{fig:attestation_probability_comparison_b}, we see that the defender incurs higher computational cost with $\vec{p} = \vec{1}$ than with the optimal strategy $\vec{p}^*$, making the latter the better choice.

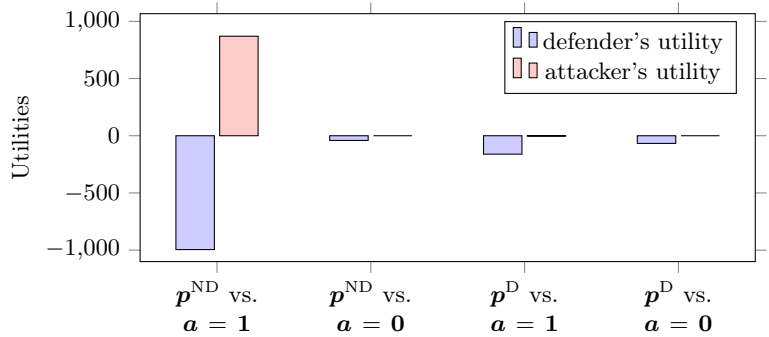
\begin{figure}[t]
\pgfplotstableread[col sep = semicolon]{attacker_defender_game.csv}\table
\centering
    \begin{tikzpicture}
    \begin{axis}[
            xbar,
            legend pos=north east,
            SmallBarPlot,
            width=0.8\textwidth,
            height=0.4\textwidth,
            xticklabels from table={\table}{case},
            x tick label style={
            },
            ylabel=Utilities,
            enlarge x limits={abs=0.5},
            ymin=-1100,
        ]
        \addplot [BlueBars] table [x expr=\coordindex, y=defender] {\table};
        \legend{defender's utility}
        \addplot [RedBars] table [x expr=\coordindex, y=attacker] {\table};
        \addlegendentry{attacker's utility}
    \end{axis}
    \end{tikzpicture}
\caption{Players' utilities in various strategy profiles: not deter vs. attack ($\vec{p}^\text{ND}$ vs. $\vec{a}=\vec{1}$),  not deter vs. not attack ($\vec{p}^\text{ND}$ vs. $\vec{a}=\vec{0}$),  deter vs. attack ($\vec{p}^\text{D}$ vs. $\vec{a}=\vec{1}$), deter vs. not attack ($\vec{p}^\text{D}$ vs. $\vec{a}=\vec{0}$).}
\label{fig:attacker_defender_game}
\end{figure}

\paragraph{Comparison of Strategy Profiles}
To better understand how the optimal attestation strategy outperforms the other strategies, we now take a closer look at the players' utilities in specific strategy profiles. 
For the defender, we consider two strategies:
optimal strategy given that the defender does not completely deter the attacker ($\vec{p}^\text{ND}$, see Proposition~\ref{prop:multi_dev_sing_class}) and optimal strategy that completely deters the attacker ($\vec{p}^\text{D}$, which is the optimal attestation strategy in this problem instance).
For the adversary, we also consider two strategies:
attacking every device ($\vec{a}=\vec{1}$, which is a best response against $\vec{p}^\text{ND}$ in this problem instance) and not attacking at all ($\vec{a}=\vec{0}$, which is always a best response against $\vec{p}^\text{D}$). 



Figure~\ref{fig:attacker_defender_game} shows the players' utilities in the four strategy profiles formed by the above strategies.
We observe that the defender's utility is highest when it does not completely deter the adversary from attacking and the adversary does not attack at all ($\vec{p}^\text{ND}$ vs. $\vec{a}=\vec{0}$) since the defender incurs minimal computational cost and suffers no security losses in this case. However, this is not an equilibrium since the adversary can attain higher utility by attacking every device (see $\vec{p}^\text{ND}$ vs. $\vec{a}=\vec{1}$), which results in the best utility for the adversary and worst for the defender.
To avoid such catastrophic losses, the defender can use the strategy of complete deterrence, in which case the adversary will be indifferent between attacking and not attacking ($\vec{p}^\text{D}$ vs. $\vec{a}=\vec{0}$ and $\vec{a}=\vec{1}$). 
Note that since the defender's utility is higher if the adversary does not attack, the defender can opt to tip the balance in favor of not attacking through an infinitesimal change.

\section{Related Work} \label{sec:related_work}

In this section, first we discuss the pros and cons of different IoT testbeds in the existing literature. Then we discuss existing works related to hardware- or software-based IoT remote attestation. Finally, we present a few SSG works and how our approach is different from theirs.




\subsection{IoT Security Testbeds}


General application- or hardware-oriented IoT testbeds are widely used for research works. The primary concern for these testbeds are to find ideal configurations or set up in different types of environments~\cite{belli2015design,adjih2015fit}. Several IoT security testbeds are available to test different security and reliability issues. For example,
Siboni et al.\ proposed an IoT security testbed framework, arguing that an ideal testbed should ensure reliability, anti-forensics, and adaptivity~\cite{siboni2019security}. In another work, 
Arseni et al.\ developed a heterogeneous IoT testbed named \emph{Pass-IoT}, consisting of three different architectures (MCU, SDSoC, traditional CPU) that can test, optimize, and develop lightweight cryptographic algorithms~\cite{arseni2016pass}.

Nowadays, the number of IoT applications is rising, and so are associated security concerns for these applications. Therefore,  
Tekeoglu et al.\ developed a security testbed that can perform privacy analysis of IoT devices, including HDMI sticks, IP cameras, smartwatches, and drones~\cite{tekeoglu2016testbed}. The testbed enables identifying insecure protocol versions, authentication issues, and privacy \mbox{violations}. 

We find the existing IoT security testbeds offering general security concerns related to cryptographic development, secure protocol implementations,  and data privacy issues. Our remote attestation testbed offers testing application vulnerabilities, developing associated exploits, and evaluating mitigation measures through software-oriented remote attestation.
Table~\ref{tab: comparison} presents a comparative analysis between our work and the existing IoT security testbeds.

\begin{table}[!ht]
\centering
\caption{IoT Security Testbeds}
\label{tab: comparison}
\resizebox{0.95\textwidth}{!}{
\begin{tabular}{|c|c|c|c|}
\hline
\diagbox[width=18em]{Features}{Works}     &
\begin{tabular}[c]{@{}c@{}}Arseni et\\ al., 2016\end{tabular} & \begin{tabular}[c]{@{}c@{}}Tekeoglu et\\ al., 2016\end{tabular} & Our Work \\ \hline\hline
\begin{tabular}[c]{@{}c@{}}Lightweight Encryption Algorithms Development \end{tabular} & \checkmark & \text{\sffamily X}  & \text{\sffamily X} \\\hline
\begin{tabular}[c]{@{}c@{}}Vulnerability Scans\end{tabular}  &  \text{\sffamily X} &  \checkmark & \checkmark \\ \hline
\begin{tabular}[c]{@{}c@{}}Authentication, Privacy Violations\end{tabular} & \text{\sffamily X} &  \checkmark &  \text{\sffamily X} \\ \hline
\begin{tabular}[c]{@{}c@{}}Exploitation Development and Analysis\end{tabular} & \text{\sffamily X} &  \text{\sffamily X} &  \checkmark \\ \hline
\begin{tabular}[c]{@{}c@{}}Mitigation Measures Testing\end{tabular} & \text{\sffamily X} &  \text{\sffamily X} &  \checkmark  \\ \hline
\begin{tabular}[c]{@{}c@{}}Remote Attestation Experiments\end{tabular} & \text{\sffamily X} &  \text{\sffamily X} &  \checkmark  \\ \hline
\end{tabular}
}
\end{table}

\subsection{Remote Attestation}

Checksum-based remote attestation has been widely used to secure IoT devices for a long time. Earlier, Seshadri et al.\ proposed a cell-based pseudo-random traversal approach in their software-based attestation scheme entitled \emph{SWATT}~\cite{seshadri2004swatt}. The authors developed the 8-bit micro-controller architecture scheme to generate random addresses to checksum using an RC4 stream cipher. Yang et al.\ proposed a distributed software-based attestation scheme for WSN to verify the integrity of code in a distributed fashion~\cite{yang2007distributed}. The works led to later works in terms of cell- and block-based pseudo-random checksum, respectively.

A few recent works include hardware-assisted remote runtime attestation~\cite{geden2019hardware} that addresses runtime attack detection, a low-cost checksum-based remote memory attestation for smart grid~\cite{yang2015towards}, lightweight remote attestation in distributed wireless sensor networks where all nodes validate each other's data~\cite{kiyomoto2014lightweight}, and so on. Survey papers on attestation, for example, Steiner et al. presented a more comprehensive overview on checksum-based attestation~\cite{steiner2016attestation}.

\subsection{Stackelberg Security Games}
SSGs have been successfully applied to security problems in resource-limited domains such as airport security, biodiversity protection, randomized inspections, border patrols, cyber security, and so on~\cite{trejo2016adapting,gan2018stackelberg,bucarey2017building}. 
To the best of our knowledge, our work is the first to apply game theory in the area of remote attestation. 

While there is no prior work within the intersection of game theory and remote attestation, a number of research efforts have applied SSGs to other detection problems that resemble ours.
%
%
For example, Wahab et al. developed a Bayesian Stackelberg game that helps the defender to determine optimal detection load distribution strategy among virtual machines within a cloud environment~\cite{wahab2019resource}. 
As another example, Chen et al. develops an SSG model that detects adversarial outbreak in an IoT environment through determining strategic dynamic scheduling of intrusion detection systems~\cite{chen2020stackelberg}.
\section{Conclusion and Future Work} \label{sec:conclusion}

\vspace{-0.25em}

IoT device exploitation has been a significant issue lately, and organizations are investing significant resources and effort into managing these security risks.
An important approach for mitigating this threat is {remote attestation}, 
    which enables the defender to remotely verify the integrity of devices and their software. 
In this work, we developed a testbed that offers research opportunities to explore and analyze IoT vulnerabilities and exploitation and to conduct experiments with varous remote attestation methods.

So far, we have developed attack strategies mostly for when kernel and compiler-based security measures are disabled. In future work, we plan to include exploitation with security features enabled in the resource-limited IoT environment. Additionally, in this paper, we evaluated software-based attestation methods (sequential and random memory-based checksum). We intend to include some other variants of attestation methods (e.g., hybrid checksum) in the testbed and to conduct experiments with control-flow integrity.

Further, we have showed how to optimize remote-attestation strategies by formulating and studying a Stackelberg security game model. 
Our analytical results provide algorithmic solutions for finding optimal attestation strategies in a variety of settings.
These results can provide guidance to practitioners on how to  protect IoT devices using remote attestation in resource-limited environments. 

In this work we discussed optimal strategies for one attestation method ($|\calM|=1$). In future, we plan to provide analytical solutions for more general cases of multiple devices ($|\calD|>1$), multiple classes ($|\calE|>1$), and multiple attestation methods ($|\calM|>1$) as well.
Additionally, we intend to refine our model and find optimal strategies in a more complex environment using machine learning algorithms (e.g., reinforcement learning).

\subsubsection*{Acknowledgments}
This material is based upon work supported by the National Science Foundation under Grant No. CNS-1850510, IIS-1905558, and ECCS-2020289 and by the Army Research Office under Grant No. W911NF1910241 and W911NF1810208.


\bibliographystyle{splncs04}
\bibliography{reference}

\begin{thebibliography}{10}
\providecommand{\url}[1]{\texttt{#1}}
\providecommand{\urlprefix}{URL }
\providecommand{\doi}[1]{https://doi.org/#1}

\bibitem{abera2016things}
Abera, T., Asokan, N., Davi, L., Koushanfar, F., Paverd, A., Sadeghi, A.R.,
  Tsudik, G.: Things, trouble, trust: on building trust in {IoT} systems. In:
  Proceedings of the 53rd Annual Design Automation Conference. pp.~1--6 (2016)

\bibitem{adjih2015fit}
Adjih, C., Baccelli, E., Fleury, E., Harter, G., Mitton, N., Noel, T.,
  Pissard-Gibollet, R., Saint-Marcel, F., Schreiner, G., Vandaele, J., et~al.:
  {FIT IoT-LAB}: A large scale open experimental {IoT} testbed. In: 2015 IEEE
  2nd World Forum on Internet of Things (WF-IoT). pp. 459--464. IEEE (2015)

\bibitem{arseni2016pass}
Arseni, {\c{S}}.C., Mi{\c{t}}oi, M., Vulpe, A.: Pass-{IoT}: A platform for
  studying security, privacy and trust in {IoT}. In: 2016 International
  Conference on Communications (COMM). pp. 261--266. IEEE (2016)

\bibitem{belli2015design}
Belli, L., Cirani, S., Davoli, L., Gorrieri, A., Mancin, M., Picone, M.,
  Ferrari, G.: Design and deployment of an {IoT} application-oriented testbed.
  Computer  \textbf{48}(9),  32--40 (2015)

\bibitem{bucarey2017building}
Bucarey, V., Casorr{\'a}n, C., Figueroa, {\'O}., Rosas, K., Navarrete, H.,
  Ord{\'o}{\~n}ez, F.: Building real {Stackelberg} security games for border
  patrols. In: International Conference on Decision and Game Theory for
  Security. pp. 193--212. Springer (2017)

\bibitem{chen2020stackelberg}
Chen, L., Wang, Z., Li, F., Guo, Y., Geng, K.: A stackelberg security game for
  adversarial outbreak detection in the internet of things. Sensors
  \textbf{20}(3), ~804 (2020)

\bibitem{gan2018stackelberg}
Gan, J., Elkind, E., Wooldridge, M.: Stackelberg security games with multiple
  uncoordinated defenders. In: Proceedings of the 17th International Conference
  on Autonomous Agents and MultiAgent Systems. pp. 703--711 (2018)

\bibitem{geden2019hardware}
Geden, M., Rasmussen, K.: Hardware-assisted remote runtime attestation for
  critical embedded systems. In: 2019 17th International Conference on Privacy,
  Security and Trust (PST). pp. 1--10. IEEE (2019)

\bibitem{kiyomoto2014lightweight}
Kiyomoto, S., Miyake, Y.: Lightweight attestation scheme for wireless sensor
  network. International Journal of Security and Its Applications
  \textbf{8}(2),  25--40 (2014)

\bibitem{nunes2019vrased}
Nunes, I.D.O., Eldefrawy, K., Rattanavipanon, N., Steiner, M., Tsudik, G.:
  Vrased: A verified hardware/software co-design for remote attestation. In:
  28th USENIX Security Symposium (USENIX Security 19). pp. 1429--1446 (2019)

\bibitem{parikh2017aslr}
Parikh, V., Mateti, P.: Aslr and rop attack mitigations for arm-based android
  devices. In: International Symposium on Security in Computing and
  Communication. pp. 350--363. Springer (2017)

\bibitem{seshadri2005pioneer}
Seshadri, A., Luk, M., Shi, E., Perrig, A., Van~Doorn, L., Khosla, P.: Pioneer:
  Verifying integrity and guaranteeing execution of code on legacy platforms.
  In: Proceedings of ACM Symposium on Operating Systems Principles (SOSP).
  vol.~173, pp. 10--1145 (2005)

\bibitem{seshadri2004swatt}
Seshadri, A., Perrig, A., Van~Doorn, L., Khosla, P.: Swatt: Software-based
  attestation for embedded devices. In: IEEE Symposium on Security and Privacy,
  2004. Proceedings. 2004. pp. 272--282. IEEE (2004)

\bibitem{siboni2019security}
Siboni, S., Sachidananda, V., Meidan, Y., Bohadana, M., Mathov, Y., Bhairav,
  S., Shabtai, A., Elovici, Y.: Security testbed for internet-of-things
  devices. IEEE Transactions on Reliability  \textbf{68}(1),  23--44 (2019)

\bibitem{sinha2018stackelberg}
Sinha, A., Fang, F., An, B., Kiekintveld, C., Tambe, M.: Stackelberg security
  games: Looking beyond a decade of success. In: Proceedings of the 27th
  International Joint Conference on Artificial Intelligence. IJCAI (2018)

\bibitem{steiner2016attestation}
Steiner, R.V., Lupu, E.: Attestation in wireless sensor networks: A survey. ACM
  Computing Surveys (CSUR)  \textbf{49}(3),  1--31 (2016)

\bibitem{tekeoglu2016testbed}
Tekeoglu, A., Tosun, A.{\c{S}}.: A testbed for security and privacy analysis of
  {IoT} devices. In: 2016 IEEE 13th International Conference on Mobile Ad Hoc
  and Sensor Systems (MASS). pp. 343--348. IEEE (2016)

\bibitem{trejo2016adapting}
Trejo, K.K., Clempner, J.B., Poznyak, A.S.: Adapting strategies to dynamic
  environments in controllable {Stackelberg} security games. In: 2016 IEEE 55th
  Conference on Decision and Control (CDC). pp. 5484--5489. IEEE (2016)

\bibitem{wahab2019resource}
Wahab, O.A., Bentahar, J., Otrok, H., Mourad, A.: Resource-aware detection and
  defense system against multi-type attacks in the cloud: Repeated bayesian
  stackelberg game. IEEE Transactions on Dependable and Secure Computing
  \textbf{18}(2),  605--622 (2019)

\bibitem{xu2018security}
Xu, B., Wang, W., Hao, Q., Zhang, Z., Du, P., Xia, T., Li, H., Wang, X.: A
  security design for the detecting of buffer overflow attacks in {IoT} device.
  IEEE Access  \textbf{6},  72862--72869 (2018)

\bibitem{yang2015towards}
Yang, X., He, X., Yu, W., Lin, J., Li, R., Yang, Q., Song, H.: Towards a
  low-cost remote memory attestation for the smart grid. Sensors
  \textbf{15}(8),  20799--20824 (2015)

\bibitem{yang2007distributed}
Yang, Y., Wang, X., Zhu, S., Cao, G.: Distributed software-based attestation
  for node compromise detection in sensor networks. In: 26th IEEE International
  Symposium on Reliable Distributed Systems (SRDS 2007). pp. 219--230. IEEE
  (2007)

\bibitem{yin2010stackelberg}
Yin, Z., Korzhyk, D., Kiekintveld, C., Conitzer, V., Tambe, M.: Stackelberg vs.
  {Nash} in security games: Interchangeability, equivalence, and uniqueness.
  In: Proceedings of the 9th International Conference on Autonomous Agents and
  Multiagent Systems: volume 1-Volume 1. pp. 1139--1146 (2010)

\end{thebibliography}

\ifExtendedVersion
\appendix 
\section{Appendix} \label{sec:proofs}

In this appendix, we present the step-by-step analysis of our Stackelberg security game between the attacker and the defender. Here, we show how the defender can take advantage of being the first to move, and how the attacker its plays best response to the defender's move.

\subsection{Case 1: Single Device and Single Attestation Method}

\subsubsection{Attacker's Best Response:}

\begin{proof}[Proof of Lemma~\ref{lem:single_device_single_class}]
Considering an individual device to attack or not, we can derive the attacker's expected utility from Eqn.~\ref{eqn:attacker_utility}. Now, from Eqn.~\ref{eqn:attacker_utility}, we see the loss $L_A^\delta$, and the gain $G_A^\delta$
are constants. Here, to maximize the expected utility, the attacker should choose the best response based on the following
\begin{align}
    &\operatorname{argmax}_{\vec a}U_A(\vec p, \vec a)\notag\\ 
    &= \operatorname{argmax}_{\vec a} \Big(\big[ L_A^\delta \cdot P_\delta(p_\delta^m) + G_A^\delta \cdot (1-P_\delta(p_\delta^m)) \big] \cdot a_\delta - C^T_A(\vec a)\Big ) . \label{eqn: attacker_max_utility}
\end{align}

If the attacker chooses to attack a single device from a particular class $\calE$, attacker needs to develop exploits for that class. Therefore, from Eqn.~\ref{eqn:attacker_cost_2} we get
\begin{align}
    C_A^T (\vec a) = C_A^{\calE} + C_A^\delta \label{eqn:attacker_case_single_device_1}
\end{align}
where $\calE \subset \calD, \delta \in \calE, |\calE| = 1$.

Now, if the attacker chooses to attack a particular device, i.e., $a_\delta=1$, the expected utility is 
\begin{align}
    U_A(\vec p, 1) = \big[ L_A^\delta \cdot P_\delta(p_\delta^m) + G_A^\delta \cdot (1-P_\delta(p_\delta^m)) \big]   - (C_A^{\calE} + C_A^\delta) .
\end{align}

In contrast, if the attacker chooses not to attack and does not develop an exploit, its expected utility is
\begin{align}
    U_A(\vec p, 0) = 0 .
\end{align}

We assume that the attacker only attacks when it expects that the utility from attacking will be greater than the utility of not attacking:
\begin{align}
    U_A(\vec p, 1) &> U_A(\vec p, 0) \notag \\
      L_A^\delta \cdot P_\delta(p_\delta^m) + G_A^\delta \cdot (1-P_\delta(p_\delta^m))  - \big (C_A^{\calE} + C_A^\delta \big ) &> 0 \notag\\
      L_A^\delta \cdot P_\delta(p_\delta^m) + G_A^\delta \cdot (1-P_\delta(p_\delta^m))  &> C_A^{\calE} + C_A^\delta \notag\\
     L_A^\delta \cdot P_\delta(p_\delta^m) - G_A^\delta\cdot P_\delta(p_\delta^m) &> C_A^{\calE} + C_A^\delta - G_A^\delta \notag\\
     \big( L_A^\delta - G_A^\delta \big) \cdot P_\delta(p_\delta^m) &> C_A^{\calE} + C_A^\delta - G_A^\delta \notag\\
     P_\delta(p_\delta^m) &< \frac{C_A^{\calE} + C_A^\delta - G_A^\delta}{L_A^\delta - G_A^\delta} .
\end{align}
Note that $L_A^\delta - G_A^\delta$ is negative by definition.

Now, from Eqn.~\ref{eqn:accuracy}, for a particular attestation method, we get
\begin{align}
    1 - (1 - \mu^m \cdot p_\delta^m) &< \frac{C_A^{\calE} + C_A^\delta - G_A^\delta}{L_A^\delta - G_A^\delta} \notag\\
     \mu^m \cdot p_\delta^m &< \frac{C_A^{\calE} + C_A^\delta - G_A^\delta}{L_A^\delta - G_A^\delta}\notag\\
     p_\delta^m &< \frac{1}{\mu^m} \cdot \frac{C_A^{\calE} + C_A^\delta - G_A^\delta}{L_A^\delta - G_A^\delta} .
\end{align}

Hence, whether the attacker chooses to attack or not depends on the following conditions
\begin{align}
    \operatorname{argmax}_{\vec a}U_A(\vec p, \vec a) = \begin{cases}
    1 & \text{ if $p_\delta^m < \frac{1}{\mu^m} \cdot \frac{C_A^{\calE} + C_A^\delta - G_A^\delta}{L_A^\delta - G_A^\delta}$} \\
    0 & \text{ otherwise.}
    \end{cases} \label{eqn:final_single_device_single_attestation_condition_appendix}
\end{align}

From Eqn.~\ref{eqn:final_single_device_single_attestation_condition_appendix}, we see the attacker chooses to attack or not based on the threshold value $\tau_\delta = \frac{1}{\mu^m} \cdot \frac{C_A^{\calE} + C_A^\delta - G_A^\delta}{L_A^\delta - G_A^\delta}$. Now, for $p_\delta^m = \tau_\delta$, the attacker's utility is not optimal regardless of whether the attacker chooses to attack or not. 
\end{proof}


\subsubsection{Defender's Optimal Strategy:}
\begin{proof}[Proof of Proposition~\ref{prop:single_device_single_class}]
The defender knows the optimal strategy of an attacker and chooses the value of the probability of attestation $p_\delta^m$ accordingly. 
First, we calculate the cost for the defender for a single attestation method:
\begin{align}
    C_D^T (\vec p) &= C_D^m \cdot p_\delta^m && \delta \in \calD, m \in \calM
\end{align}

Now, the defender knows that the attacker's strategy is either attack or no-attack. Therefore, the defender calculates utilities for these two conditions.

\paragraph{Defender chooses not to deter the attacker:}
If the attacker chooses to attack, from Eqn.~\ref{eqn:defender_utility}, the defender's utility is
\begin{align}
    U_D(\vec p,1)
    &=  G^\delta_D \cdot P_\delta(p_\delta^m) + L^\delta_D \cdot (1-P_\delta(p_\delta^m)) - C_D^m \cdot p_\delta^m\notag\\
    &= G^\delta_D \cdot p_\delta^m \cdot \mu^m + L^\delta_D \cdot (1-p_\delta^m \cdot \mu^m ) - C_D^m \cdot p_\delta^m \notag\\
    &= p_\delta^m \cdot (G^\delta_D \cdot \mu^m - L^\delta_D \cdot \mu^m - C_D^m) + L^\delta_D 
    \label{eqn:defender_utility_single_device_attack}
\end{align}
where $p_\delta^m = [0,\tau_\delta)$.

If $C_D^m < (G^\delta_D  - L^\delta_D)\cdot \mu^m$, the defender chooses $p_\delta^m = \tau_\delta$ for the maximum utility:
\begin{align}
    \max_{\vec p, \vec a \in \calF_A(\vec p)} U_D(\vec p,1) = \tau_\delta \cdot (G^\delta_D \cdot \mu^m - L^\delta_D \cdot \mu^m - C_D^m) + L^\delta_D .
    \label{eqn:defender_utility_tau_not_zero}
\end{align}

Otherwise, the defender chooses a $p_\delta^m = 0$ for the maximum utility:
\begin{align}
    \max_{\vec p, \vec a \in \calF_A(\vec p)} U_D(\vec p,1) = L_D^\delta .
    \label{eqn:defender_utility_tau_zero}
\end{align}

\paragraph{Defender chooses to deter the attacker:}
If the attacker chooses not to attack, the defender's utility is
\begin{align}
    U_D(\vec p,0) = 0 - C_D^m \cdot p_\delta^m \label{eqn:defender_utility_single_device_no_attack}
\end{align}
where $p_\delta^m = [\tau_\delta,1]$.

Here, $\tau_\delta < 1$ and Eqn.~\ref{eqn:defender_utility_single_device_no_attack} attains its maximum for $p_\delta^m = \tau_\delta$. Therefore, the defender's maximum utility is
\begin{align}
    \max_{\vec p, \vec a \in \calF_A(\vec p)} U_D(\vec p,0) = - C_D^m \cdot \tau_\delta .
       \label{eqn:attacker_choose_not_to_attack}
\end{align}

\paragraph{Utility Comparison:}
The defender does not have an attestation cost if there is no attack. So, the maximum utility for the defender is greater if the attacker chooses not to attack than an attack.

\begin{align}
   \max_{\vec p, \vec a \in \calF_A(\vec p)} U_D(\vec p,0) &> \max_{\vec p, \vec a \in \calF_A(\vec p)} U_D(\vec p,1) .
   \label{eqn:utility_comparison}
\end{align}

Now, if the defender chooses $p_\delta^m = \tau_\delta$, then we get from Eqn.~\ref{eqn:defender_utility_tau_not_zero} and Eqn.~\ref{eqn:attacker_choose_not_to_attack} that
\begin{align}
     - C_D^m \cdot \tau_\delta &> \tau_\delta \cdot (G^\delta_D \cdot \mu^m - L^\delta_D \cdot \mu^m - C_D^m) + L^\delta_D \notag\\
     \tau_\delta \cdot \mu^m \cdot (L^\delta_D - G^\delta_D) &> L_D^\delta \notag\\
     \tau_\delta &< \frac{L_D^\delta}{\mu^m \cdot (L^\delta_D - G^\delta_D)}  .
\end{align}

On the other hand, if the defender chooses $p_\delta^m = 0$, then we get from Eqn.~\ref{eqn:defender_utility_tau_zero} and Eqn.~\ref{eqn:attacker_choose_not_to_attack} that
\begin{align}
     - C_D^m \cdot \tau_\delta &> L^\delta_D \notag\\
     \tau_\delta &< \frac{L_D^\delta}{- C_D^m}  .
\end{align}

Hence, the probability of attestation chosen by the defender depends on the following conditions.
%
If $C_D^m < (G^\delta_D  - L^\delta_D)\cdot \mu^m$, the defender's choice is
\begin{align}
    \operatorname{argmax}_{\vec p,\, \vec a \in \calF_A(\vec p)} U_D(\vec p, \vec a)
    = \tau_\delta 
\end{align}
Otherwise, if $C_D^m \geq (G^\delta_D  - L^\delta_D)\cdot \mu^m$, the defender's choice is
\begin{align}
    \operatorname{argmax}_{\vec p,\, \vec a \in \calF_A(\vec p)} U_D(\vec p, \vec a)
    = \begin{cases} \tau_\delta & \text{if } \tau_\delta < \frac{L_D^\delta}{- C_D^m}\\
    0 & \text{otherwise.} \end{cases}
\end{align}
\end{proof}

\subsection{Case 2: Multiple Devices and Single Device Class}

\subsubsection{Attacker's Best Response:}
\begin{proof}[Proof of Lemma~\ref{lem:multi_dev_sing_class}]
If the attacker chooses to attack multiple devices from a particular class $\calE$, attacker needs to develop exploits only once for the class. Now, while choosing devices to attack, the attacker has two cases: to attack a device $\delta \in \calE$, attacker has to develop exploit for that class, and to attack additional devices from the same class the attacker does not have additional cost for developing exploits.

Now, if there is no additional cost to develop exploits, to attack additional device $\delta \in \calD$ from the same class is 
\begin{align}
    C_A^T (1) = C_A^\delta .
\end{align}

Therefore, from Eqn.~\ref{eqn:attacker_case_single_device_1} the total cost is
\begin{align}
    C_A^T (\vec a) 
    &= C_A^{\calE} + \sum_{\delta \in \calE} C_A^\delta \cdot a_\delta \label{eqn:single_class_cost_att}
\end{align}
where $\calE \subset \calD, |\calE| > 1$.

\paragraph{Case 1: Attacker chooses to attack some devices:}
From Eqn.~\ref{eqn:final_single_device_single_attestation_condition_appendix}, the attacker can choose to attack a device from the class if $p_\delta^m < \frac{1}{\mu^m} \cdot \frac{C_A^{\calE} + C_A^\delta - G_A^\delta}{L_A^\delta - G_A^\delta}$. 
As the attacker has no cost for developing exploits for the rest of the devices, attacker can choose to attack other devices if $p_\delta^m < \frac{1}{\mu^m} \cdot \frac{C_A^\delta - G_A^\delta}{L_A^\delta - G_A^\delta}$ for $\delta \in \calE$, and hence, $\tau_\delta = \frac{1}{\mu^m} \cdot \frac{C_A^\delta - G_A^\delta}{L_A^\delta - G_A^\delta}$. Here, the set of the attacker's chosen devices from a single class is $\calA$, where 

\begin{align}
    \calA = \left\{\delta \in \calE ~\middle|~ p_\delta < \frac{1}{\mu^m} \cdot \frac{C_A^\delta - G_A^\delta}{L_A^\delta - G_A^\delta} \right\} . \label{eqn:chosen_class_A_definition} 
\end{align}

Now, if the attacker chooses to attack the chosen devices $\delta \in \calA$, from Eqn.~\ref{eqn:single_class_cost_att}, the total cost is

\begin{align}
    C_A^T (\vec a) 
    &= C_A^{\calE} + \sum_{\delta \in \calA} C_A^\delta && \text{where }\, a_\delta
    = \begin{cases} 1 & \delta \in \calA\\
    0 & \text{otherwise.} \end{cases} \label{eqn:attacker_case_multiple_device}
\end{align}

Hence, the expected utility is for attacking the chosen devices is

\begin{align}
    U_A(\vec p, \vec a) = \sum_{\delta \in \calA} \big [ L_A^\delta \cdot P_\delta(p_\delta^m) + G_A^\delta \cdot (1-P_\delta(p_\delta^m)) \big ] - C_A^{\calE}  - \sum_{\delta \in \calA} C_A^\delta\\
    \text{where }\, a_\delta
    = \begin{cases} 1 & \delta \in \calA\\
    0 & \text{otherwise.} \end{cases} \notag
\end{align}

\paragraph{Case 2: Attacker chooses not to attack any device:}
If the attacker chooses not to attack any device, the total cost is 
\begin{align}
    C_A^T (\vec 0) = 0 && \text{where $\delta \in \calD$ and $a_\delta = 0$.}
\end{align}

Now, the expected utility of the attacker is

\begin{align}
    U_A(\vec p, \vec 0) = 0. && \text{where $\delta \in \calD$ and $a_\delta = 0$.}
\end{align}

\paragraph{Utility Comparison:}
Consider the expected utility of attacking some devices is greater than the expected utility if the attacker chooses not to attack:
\begin{align}
    U_A(\vec p, \vec a) > U_A(\vec p, \vec 0) \notag \\
    \text{where $\delta \in \calA ( a_\delta = 1 )$} \notag
\end{align}

Hence,
\begin{align}
     \sum_{\delta \in \calA} L_A^\delta \cdot P_\delta(p_\delta^m) + G_A^\delta \cdot (1-P_\delta(p_\delta^m))   - C_A^{\calE} - \sum_{\delta \in \calA} C_A^\delta > 0 \notag\\
      \sum_{\delta \in \calA} \big [ L_A^\delta \cdot P_\delta(p_\delta^m) \big ] + \sum_{\delta \in \calA} \big [ G_A^\delta \cdot (1-P_\delta(p_\delta^m)) \big ]  > C_A^{\calE} + \sum_{\delta \in \calA} C_A^\delta \notag\\
     \sum_{\delta \in \calA} \big [ L_A^\delta \cdot P_\delta(p_\delta^m) \big ] - \sum_{\delta \in \calA} \big [ G_A^\delta\cdot P_\delta(p_\delta^m) \big ]  > C_A^{\calE} + \sum_{\delta \in \calA} C_A^\delta - \sum_{\delta \in \calA}  G_A^\delta \notag\\
     \sum_{\delta \in \calA} \big [ \big( L_A^\delta - G_A^\delta \big) \cdot P_\delta(p_\delta^m)  \big ]
    > C_A^{\calE} + \sum_{\delta \in \calA} \big( C_A^\delta -  G_A^\delta \big) . \label{eqn: attacker_multiple_device_single_class_utility_comparison}
\end{align}

Now, from Eqn.~\ref{eqn:accuracy}, for a single attestation method, $P_\delta(p_\delta^m) = p_\delta^m \cdot \mu^m$. Therefore, we derive from Eqn.~\ref{eqn: attacker_multiple_device_single_class_utility_comparison} as follows 


\begin{align}
    \sum_{\delta \in \calA} \big [ \big( L_A^\delta - G_A^\delta \big) \cdot p_\delta^m \cdot \mu^m  \big ]
    > C_A^{\calE} + \sum_{\delta \in \calA} \big( C_A^\delta -  G_A^\delta \big) \notag\\
     \mu^m \cdot \sum_{\delta \in \calA} \big [ \big( L_A^\delta - G_A^\delta \big) \cdot p_\delta^m \big ]
    > C_A^{\calE} + \sum_{\delta \in \calA} \big( C_A^\delta -  G_A^\delta \big) . \label{eqn:attacker_multi_device_single_attest}
\end{align}

Now, for each selected device $\delta \in \calA$, the best response of the attacker $\vec a \in \calF_A(\vec p)$ depends on the following conditions 
\begin{align}
    a_\delta = \begin{cases}
    1 & \text{ if Eqn.~\ref{eqn:attacker_multi_device_single_attest} is true and } \delta \in \mathcal{A} \\
    0 & \text{ otherwise.}
    \end{cases}  \label{eqn:multiple_device_single_attestation_condition}
\end{align}

Given the defender's attestation strategy $\vec{p}$, let 
\begin{equation}
{a}^*_\delta = \begin{cases} 1 & \text{ if } p_\delta^m < \overline{\tau}_\delta \\
0 & \text{ otherwise,}
\end{cases}
\end{equation}
where
\begin{equation}
  \overline{\tau}_\delta =  \frac{1}{\mu^m} \cdot \frac{C_A^\delta - G_A^\delta}{L_A^\delta - G_A^\delta} 
\end{equation}
If $U_A\left(\vec{p}, \vec{a}^*\right) \geq 0$, then $\vec{a}^*$ is a best-response strategy for the attacker; otherwise, the only best-response strategy is $\vec{a} = \vec{0}$.

\end{proof}

\subsubsection{Defender's Optimal Strategy:}
\begin{proof}[Proof of Proposition~\ref{prop:multi_dev_sing_class}]
The defender knows that the attacker's strategy is either to attack some devices $\delta \in \calA$ from a particular class or not to attack at all. Therefore, the defender calculates optimal utilities for the following two conditions.

\paragraph{Defender chooses not to deter the attacker:}

The attacker selects the set of chosen devices $\calA$ based on the defender's deployed probability of attestation $p_\delta^m$ value. 
Here, the defender's optimal expected utility for the devices $\delta \in \calE$ is
\begin{align*}
    &\max_{\vec p, \vec a \in \calF_A(\vec p)}\ U_D(\vec p, \vec a) = \sum_{\delta \in \calE} \left ( \big [ G^\delta_D \cdot P_\delta(p_\delta^m) + L^\delta_D \cdot (1-P_\delta(p_\delta^m)) \big ] \cdot a_\delta - C_D^m \cdot p_\delta^m \right )
\end{align*}
where $p_\delta^m \in [0,\tau_\delta]$.

Now, the defender needs to calculate utility for individual devices as the attacker's decision entirely depends on the combined expected utilities from the chosen devices. Here, the attacker chooses to attack devices if the probability of attestation is less than the device threshold value, i.e., $p_\delta^m < \tau_\delta = \frac{1}{\mu^m} \cdot \frac{C_A^\delta - G_A^\delta}{L_A^\delta - G_A^\delta}$. Now, if the attacker chooses to attack, the defender's expected utility $U_{D}^\delta(p_\delta^m, 1)$ for a device $\delta \in \calE$ is
\begin{align}
    U_{D}^\delta(p_\delta^m, 1)
    &=  G^\delta_D \cdot P_\delta(p_\delta^m) + L^\delta_D \cdot (1-P_\delta(p_\delta^m)) -  C_D^m \cdot p_\delta^m\notag\\
    &=  G^\delta_D \cdot p_\delta^m \cdot \mu^m +  L^\delta_D \cdot (1-p_\delta^m \cdot \mu^m ) -  C_D^m \cdot p_\delta^m \notag\\
    &=  \big[ p_\delta^m \cdot (G^\delta_D \cdot \mu^m - L^\delta_D \cdot \mu^m - C_D^m)\big] +  L^\delta_D 
\end{align}
where $\delta \in \calE, p_\delta^m \in [0,\tau_\delta]$.

Here, the defender can reduce the probability of attestation if the cost of attestation is greater than the defender's loss.
Therefore, the defender's optimal choice is the following.

If $C_D^m < (G^\delta_D  - L^\delta_D)\cdot \mu^m$, then the defender's choice is
\begin{align}
    \operatorname{argmax}_{\vec p,\, \vec a \in \calF_A(\vec p)} U_D(\vec p, \vec a)
    = \tau_\delta .
\end{align}

Otherwise, if $C_D^m \geq (G^\delta_D  - L^\delta_D)\cdot \mu^m$, the defender's choice is
\begin{align}
    \operatorname{argmax}_{\vec p,\, \vec a \in \calF_A(\vec p)} U_D(\vec p, \vec a)
    = \begin{cases} \tau_\delta & \text{if } \tau_\delta < \frac{L_D^\delta}{- C_D^m}\\
    0 & \text{otherwise}. \end{cases}
\end{align}

\paragraph{Defender chooses to deter the attacker:}

Here, the defender's cost increases with the increase of the probability of attestation $p_\delta^m$. However, the defender can choose $p_\delta^m = \tau_\delta$ for all the devices to deter the attacker from attacking. Now, from the defender's perspective, not each device is equally important for the attacker or the cost of attestation is greater than the defender's loss. Therefore, the defender can choose $p_\delta^m < \tau_\delta$ for those devices.
\begin{align}
    \max_{\vec p:~ \forall \vec a \left(U_A(\vec p, \vec a) \leq 0\right)}  U_D(\vec p, \vec a) 
    &= -\sum_{\delta \in \calE} C_D^m \cdot p_\delta^m
\end{align}
where $p_\delta^m \in [0,\tau_\delta]$.

Since $p_\delta^m \in [0,\tau_\delta]$, if the attacker chooses to attack,  attacking every device is optimal.
Hence, we can consider $\vec a = \vec 1$:
\begin{align}
    \max_{\vec p}  U_D(\vec p, \vec 1) 
    &= -\sum_{\delta \in \calE} C_D^m \cdot p_\delta^m\\
    \text{such that }&\,\, U_A(\vec p, \vec 1) \leq 0 \nonumber \\
    \text{where }&\,\,  p_\delta^m \in [0,\tau_\delta] .\notag
\end{align}

\end{proof}

\subsection{Case 3: Multiple Devices and Multiple Device Classes}

\subsubsection{Attacker's Best Response}

\begin{proof}[Proof of Lemma~\ref{lem:multiple_device_multiple_class}]
If the attacker chooses to attack multiple devices from multiple classes, attacker needs to develop exploits for target classes. Now, the attacker has two choices: either to attack a set of multiple devices from multiple classes or not attack at all.

\paragraph{Case 1: Attacker chooses to attack some devices:}
If the chosen set of devices is $\calA$ (from Eqn.~\ref{eqn:chosen_class_A_definition}) from each class $\calE$
the attacker's total cost is
\begin{align}
    C_A^T (\vec a) &= \sum_{\calE} \big( C_A^{\calE} + \sum_{\delta \in \calA} C_A^\delta \big ) && \text{where $\delta \in \calA, \calA \subset \calE$}\label{eqn:attacker_case_multiple_device_multiple_class}\\
    \text{and }&\, a_\delta
    = \begin{cases} 1 & \delta \in \calA\\
    0 & \text{otherwise.} \end{cases} \notag
\end{align}

Now, if the attacker chooses to attack devices $\delta \in \calA$ from each class $\calE$, the expected utility is 
\begin{align}
    U_A(\vec p, \vec a) = \sum_{\calE} \sum_{\delta \in \calA}& \big[ L_A^\delta \cdot P_\delta(p_\delta^m) + G_A^\delta \cdot (1-P_\delta(p_\delta^m)) \big]  - \sum_{\calE} \big( C_A^{\calE} + \sum_{\delta \in \calA} C_A^\delta \big ) \\
    \text{where }&\, a_\delta
    = \begin{cases} 1 & \delta \in \calA\\
    0 & \text{otherwise.} \end{cases} \notag
\end{align}

\paragraph{Case 2: Attacker chooses not to attack any device:}
If the attacker chooses not to attack any device, i.e., $a_\delta = 0$, the total cost is
\begin{align}
    C_A^T (\vec 0) = 0 && \text{where $\delta \in \calD$ and $a_\delta = 0$,}
\end{align}
and thus the expected utility of the attacker is
\begin{align}
    U_A(\vec p, \vec 0) = 0 && \text{where $\delta \in \calD$ and $a_\delta = 0$.}
\end{align}

\paragraph{Utility Comparison:}
The expected utility of attacking a set of devices from multiple classes is greater than the expected utility if the attacker chooses not to attack if

\begin{align}
    U_A(\vec p, \vec a) > U_A(\vec p, \vec 0) \notag \\
    \text{where $\delta \in \calA ( a_\delta = 1 )$,}\, \calA \subset \calE \notag
\end{align}

Hence,
\begin{align}
     \sum_{\calE} \sum_{\delta \in \calA} \big[ L_A^\delta \cdot P_\delta(p_\delta^m) + G_A^\delta \cdot (1-P_\delta(p_\delta^m)) \big]   - \sum_{\calE} \big( C_A^{\calE} + \sum_{\delta \in \calE} C_A^\delta \big ) > 0 \notag\\
      \sum_{\calE}\sum_{\delta \in \calA} \big [ L_A^\delta \cdot P_\delta(p_\delta^m) \big ] + \sum_{\calE}\sum_{\delta \in \calA} \big [ G_A^\delta \cdot (1-P_\delta(p_\delta^m)) \big ]  > \sum_{\calE} \big( C_A^{\calE} + \sum_{\delta \in \calA} C_A^\delta \big ) \notag\\
     \sum_{\calE}\sum_{\delta \in \calA} \big [ L_A^\delta \cdot P_\delta(p_\delta^m) \big ] - \sum_{\calE}\sum_{\delta \in \calA} \big [ G_A^\delta\cdot P_\delta(p_\delta^m) \big ]  > \sum_{\calE} \big( C_A^{\calE} + \sum_{\delta \in \calA} C_A^\delta \big ) - \sum_{\calE}\sum_{\delta \in \calA}G_A^\delta \notag\\
     \sum_{\calE}\sum_{\delta \in \calA} \big [ \big( L_A^\delta - G_A^\delta \big) \cdot P_\delta(p_\delta^m) \big ] 
    > \sum_{\calE} \big( C_A^{\calE} + \sum_{\delta \in \calA} C_A^\delta \big ) - \sum_{\calE}\sum_{\delta \in \calA}G_A^\delta \notag\\
     \sum_{\calE}\sum_{\delta \in \calA} \big [ \big( L_A^\delta - G_A^\delta \big) \cdot P_\delta(p_\delta^m) \big ] > \sum_{\calE} C_A^{\calE} - \sum_{\calE}\sum_{\delta \in \calA} \big(C_A^\delta - G_A^\delta \big ) . \label{eqn: attacker_multiple_device_multiple_classes_utility_comparison}
\end{align}

Now, from Eqn.~\ref{eqn:accuracy}, for a single attestation method, $P_\delta(p_\delta^m) = p_\delta^m \cdot \mu^m$. Therefore, we derive from Eqn.~\ref{eqn: attacker_multiple_device_multiple_classes_utility_comparison} as follows
\begin{align}
    \sum_{\calE}\sum_{\delta \in \calA} \big [ \big( L_A^\delta - G_A^\delta \big) \cdot p_\delta^m \cdot \mu^m \big ] > \sum_{\calE} C_A^{\calE} - \sum_{\calE}\sum_{\delta \in \calA} \big(C_A^\delta - G_A^\delta \big ) \notag\\
     \mu^m \cdot \sum_{\calE}\sum_{\delta \in \calA} \big [ \big( L_A^\delta - G_A^\delta \big) \cdot p_\delta^m \big ] > \sum_{\calE} C_A^{\calE} - \sum_{\calE}\sum_{\delta \in \calA} \big(C_A^\delta - G_A^\delta \big ) . \label{eqn:attacker_multi_device_multi_attest}
\end{align}

Now, for each selected device $\delta \in \calA, \calA \subset \calE$, the best response of the attacker $\vec a \in \calF_A(\vec p)$ depends on the following conditions
\begin{align}
    a_\delta = \begin{cases}
    1 & \text{ if Eqn.~\ref{eqn:attacker_multi_device_multi_attest} is true and $\delta \in \calA, \calA \subset \calE$} \\
    0 & \text{ otherwise.}
    \end{cases}
\end{align}


\end{proof}

\subsubsection{Defender's Optimal Strategy:}
\begin{proof}[Proof of Proposition~\ref{prop:multiple_device_multiple_class}]
The defender knows that the attacker's strategy is either to attack some devices $\delta \in \calA$ from a particular class or not to attack at all. Therefore, the defender calculates optimal utilities for the following two conditions.

\paragraph{Defender chooses not to deter the attacker:}

The attacker selects the set of chosen devices $\calA$ based on the defender's deployed probability of attestation $p_\delta^m$ value. 
Here, the defender's optimal expected utility for the devices $\delta \in \calE$ is
\begin{align}
    &\max_{\vec p, \vec a \in \calF_A(\vec p)}\ U_D(\vec p, \vec a) \notag\\
    &= \sum_{\calE}\sum_{\delta \in \calE} \left ( \big [ G^\delta_D \cdot P_\delta(p_\delta^m) + L^\delta_D \cdot (1-P_\delta(p_\delta^m)) \big ] \cdot a_\delta - C_D^m \cdot p_\delta^m \right )    
\end{align}
where $\delta \in \calE, \calE \subset \calD,  p_\delta^m \in [0,\tau_\delta)$.

Now, the defender needs to calculate utility for individual devices as the attacker's decision entirely depends on the combined expected utilities from the chosen devices. Here, the attacker chooses to attack devices if the probability of attestation is less than the device threshold value, i.e., $p_\delta^m < \tau_\delta$. Now, if the attacker chooses to attack, the defender's expected utility $U_{D}^\delta(p_\delta^m, 1)$ of a device $\delta \in \calE$ is
\begin{align}
    U_{D}^\delta(p_\delta^m, 1)
    &=  G^\delta_D \cdot P_\delta(p_\delta^m) + L^\delta_D \cdot (1-P_\delta(p_\delta^m)) -  C_D^m \cdot p_\delta^m\notag\\
    &=  G^\delta_D \cdot p_\delta^m \cdot \mu^m +  L^\delta_D \cdot (1-p_\delta^m \cdot \mu^m ) -  C_D^m \cdot p_\delta^m \notag\\
    &=  \big[ p_\delta^m \cdot (G^\delta_D \cdot \mu^m - L^\delta_D \cdot \mu^m - C_D^m)\big] +  L^\delta_D 
\end{align}
where $\delta \in \calE, \calE \subset \calD, p_\delta^m \in [0,\tau_\delta)$.

Here, the defender can reduce the probability of attestation if the cost of attestation is greater than the defender's loss.
Therefore, the defender's optimal choice is the following.

If $C_D^m < (G^\delta_D  - L^\delta_D)\cdot \mu^m$, the defender's choice is
\begin{align}
    \operatorname{argmax}_{\vec p,\, \vec a \in \calF_A(\vec p)} U_D(\vec p, \vec a)
    = \tau_\delta .
\end{align}

Otherwise, if $C_D^m \geq (G^\delta_D  - L^\delta_D)\cdot \mu^m$, the defender's choice is
\begin{align}
    \operatorname{argmax}_{\vec p,\, \vec a \in \calF_A(\vec p)} U_D(\vec p, \vec a)
    = \begin{cases} \tau_\delta & \text{if } \tau_\delta < \frac{L_D^\delta}{- C_D^m}\\
    0 & \text{otherwise.} \end{cases}
\end{align}



\paragraph{Defender chooses to deter the attacker:}

Here, the defender's cost increases with the increase of the probability of attestation $p_\delta^m$. However, the defender can choose $p_\delta^m = \tau_\delta$ for all the devices to deter the attacker from attacking. However, not each device is equally important for the attacker or the cost of attestation is greater than the defender's loss. Therefore, the defender can choose $p_\delta^m < \tau_\delta$ for those devices:
\begin{align}
    \max_{\vec p, \vec a \in \calF_A(\vec p)}  U_D(\vec p, \vec a) 
    &= -\sum_{\calE}\sum_{\delta \in \calE} C_D^m \cdot p_\delta^m\\
    \text{where }&\,\, \delta \in \calE, \calE \subset \calD,  p_\delta^m \in [0,\tau_\delta] . \notag
\end{align}

Since $p_\delta^m \in [0,\tau_\delta]$, if the attacker chooses to attack, then attacking every device is optimal.
Hence, we can consider $\vec a = \vec 1$
\begin{align}
    \max_{\vec p}  U_D(\vec p, \vec 1) 
    &= -\sum_\calE\sum_{\delta \in \calE} C_D^m \cdot p_\delta^m\\
    \text{such that }&\,\, U_A(\vec p, \vec 1) \leq 0 \nonumber \\
    \text{where }&\,\,  p_\delta^m \in [0,\tau_\delta] . \notag
\end{align}

\end{proof}
\fi

\end{document}